\documentclass[12pt, draftclsnofoot, onecolumn,romanappendices,a4paper]{IEEEtran}
\DeclareMathSizes{12}{11}{7.7}{5.5}
\usepackage{graphicx}
\usepackage{epstopdf}
\usepackage{mathrsfs}
\usepackage{amsmath}
\usepackage{cite}
\usepackage{amsthm,amssymb,amsfonts,latexsym,bm}
\usepackage{enumerate}
\begin{document}
\makeatletter
\renewcommand{\maketag@@@}[1]{\hbox{\m@th\normalsize\normalfont#1}}
\makeatother
\title{On the Degrees of Freedom of the Symmetric Multi-Relay MIMO Y Channel}
\author{Tian~Ding, 
        Xiaojun~Yuan,~\IEEEmembership{Senior Member,~IEEE,}
        and~Soung~Chang~Liew,~\IEEEmembership{Fellow,~IEEE}
\thanks{This work has been submitted to the IEEE for possible publication.  Copyright may be transferred without notice, after which this version may no longer be accessible.}
\thanks{
T. Ding and S.C. Liew are with the Department of Information Engineering, the Chinese University of Hong Kong. X. Yuan is with the School of Information Science and Technology, ShanghaiTech University, Shanghai.}}
\maketitle
\IEEEpeerreviewmaketitle

\begin{abstract}
In this paper, we study the degrees of freedom (DoF) of the symmetric multi-relay multiple-input multiple-output (MIMO) Y channel, where three user nodes, each with $M$ antennas, communicate via $K$ geographically separated relay nodes, each with $N$ antennas. For this model, we establish a general DoF achievability framework based on linear precoding and post-processing methods. The framework poses a nonlinear problem with respect to user precoders, user post-processors and relay precoders. To solve this problem, we adopt an \textit{uplink-downlink asymmetric} strategy, where the user precoders are designed for signal alignment and the user post-processors are used for interference neutralization. With the user precoder and post-processor designs fixed as such, the original problem then reduces to a problem of relay precoder design. To address the solvability of the system, we propose a general method for solving matrix equations. Together with the techniques of antenna disablement and symbol extension, an achievable DoF of the considered model is derived for an arbitrary setup of $\left(K,M,N\right)$. We show that for $K\geq 2$, the optimal DoF is achieved for $\frac{M}{N} \in \Big[0,\max\left\{\frac{\sqrt{3K}}{3},1\right\}\Big)\cup\Big[\frac{3K+\sqrt{9K^2-12K}}{6},\infty\Big)$. We also show that the uplink-downlink asymmetric design proposed in this paper considerably outperforms the conventional approach based on uplink-downlink symmetry.
\end{abstract}


\section{Introduction}
Various wireless relaying techniques have been extensively studied for decades due to their capability to extend the coverage and enhance the capacity of wireless networks \cite{relay1,relay12,relay2,PNC}. In particular, two-way relaying based on physical-layer network coding (PNC) has attracted much research interest in the past decade \cite{PNC,TWRC1,TWRC2,TWRC3, TWRC4}. In the two-way relay channel, two users exchange information via a single relay node. Compared with conventional one-way relaying, PNC potentially doubles the spectral efficiency by allowing a relay node to decode and forward message combinations rather than individual messages. Later, the idea of PNC was extended to support efficient communications over multiway relay channels (mRC) \cite{2}, where multiple users exchange data with the help of a single relay. Efficient PNC design has been studied for various data exchange models, including pairwise data exchange \cite{MIMO3,Xchannel}, full data exchange \cite{MIMO3,Gao}, and clustered pairwise/full data exchange \cite{MIMO3,MIMO2, multirelay, yuan}. Multiple-input multiple-output (MIMO) techniques have also been incorporated into PNC-aided relay networks to achieve spatial multiplexing \cite{MIMO1}.

The capacity of the MIMO mRC generally remains a challenging open problem \cite{capacity1,capacity4}. Existing work \cite{DOF1,DOF2,DOF3,3,32,33,review1,review2} was mostly focused on analyzing the degrees of freedom (DoF) that characterizes the capacity slope at high signal-to-noise ratio (SNR). Various signaling techniques have been developed to intelligently manipulate signals and interference based on the ideas of PNC and interference alignment \cite{Jafar}. Particularly, the authors in \cite{3,32,33} studied the DoF of the MIMO Y channel, where three users exchange data in a pairwise manner with the help of a single relay. To derive the DoF of this model, a key difficulty is how to jointly optimize the linear processors, including the precoders at the user transmitters, the precoder at the relay, and the post-processers at user receivers. This problem was elegantly solved in \cite{32} by optimal design of the signal space seen at the relay, where the user precoders and post-processors are constructed by \textit{pairwise signal alignment} and \textit{uplink-downlink symmetry}, and the relay precoder by appropriate orthogonal projections. Similar ideas have also been used to derive the DoF of other multiway relay models \cite{multirelay, capacity1}.

In the work on MIMO mRC mentioned above, a major limitation is that a single relay node is employed to serve multiple user nodes simultaneously. This implies that the relay node is usually the performance bottleneck of the overall network \cite{multirelay,yuan}. As such, some recent work began to explore the potential of deploying more relay nodes for enhancing the network capacity. For instance, the authors in \cite{Lee} derived an achievable DoF of the two-relay MIMO mRC in which two pairs of users exchange messages in a pairwise manner via two relays. Later, the work in \cite{capacity2} improved the DoF result in \cite{Lee} by using the techniques of pairwise signal alignment and uplink-downlink symmetric design. The extension to the case of more than two user pairs was also considered in \cite{capacity2}. However, the DoF characterization of the multi-relay MIMO mRC is still at a very initial stage. The reason is twofold. First, for a multi-relay mRC, the relays are geographically separated and hence cannot jointly process their received signals. This implies that manipulating the relay signal space is far more difficult than that in the single-relay case. Although some of the existing techniques for single-relay mRCs can be directly borrowed for signaling design in a multi-relay mRC, the efficiency of these techniques is no longer guaranteed. Second, for given user precoders and post-processors, the solvability problem for a MIMO mRC (with single or multiple relays) can be converted to a linear system with certain rank constraints. A substantial difference between the single-relay and multiple-relay MIMO mRCs is that the linear system for the multi-relay involves multiple matrix variables, and so solving the corresponding achievability problem is much more challenging. For example, the MIMO multipair two-way relay channel with two relay nodes was considered in \cite{capacity2}. The achievability proof therein relies on some recent progresses on the solvability of linear matrix systems, and is difficult to be extended to the case with more than two relays or to other multi-relay mRCs. 

In this paper, we analyze the DoF of the symmetric multi-relay MIMO Y channel, where three user nodes, each with $M$ antennas, communicate with each other via $K$ relay nodes, each with $N$ antennas. Compared with the MIMO Y channel in \cite{3}, a critical difference is that our new model contains an arbitrary number of relays, rather than only a single relay. Following \cite{capacity2}, we formulate a general DoF achievability problem for the multi-relay MIMO Y channel based on linear processing techniques, involving the design of user precoders, relay precoders, and user post-processors. The main contributions of this paper are as follows.
\begin{itemize}
\item
In contrast to the conventional uplink-downlink symmetric design which is widely used in single-relay MIMO mRCs, we propose a new uplink-downlink asymmetric approach to solve the DoF achievability problem of the symmetric multi-relay MIMO Y channel. Specifically, in our approach, only user precoders are designed based on signal space alignment; the user post-processors are designed directly for interference neutralization. Furthermore, we show that under certain conditions, the uplink-downlink asymmetry allows the relays to deactivate a portion of receiving (not transmitting) antennas to facilitate the signal space alignment at relays. This implies that under certain conditions, some of the receiving antennas at the relays are redundant to achieve the derived DoF.
\item
Given the designed user precoders and post-processors, the original problem boils down to a linear system on the relay precoders with certain rank constraints. Due to the presence of multiple relays, the linear system involves multiple matrix variables. To tackle the solvability of this system, we establish a new technique to solve linear matrix equations with rank constraints. We emphasize that this technique can potentially be used to analyze the DoF of other multi-relay MIMO mRCs with various data exchange models, e.g., pairwise data exchange \cite{Xchannel} and clustered full data exchange \cite{MIMO2, multirelay, yuan}.
\item
Based on the above new techniques, we derive an achievable DoF of the symmetric multi-relay MIMO Y channel with an arbitrary configuration of $\left(M,N,K\right)$. Our achievable DoF is considerably higher than that derived by the conventional uplink-downlink symmetric approach. Also, a DoF upper bound is presented by assuming full cooperation among the relays and treating the multiple relays together as a single large relay. We establish the optimality of our achievable DoF for $\frac{M}{N} \in \Big[0,\max\left\{\frac{\sqrt{3K}}{3},1\right\}\Big) \cup\Big[\frac{3K+\sqrt{9K^2-12K}}{6},\infty\Big)$ by showing that the achieved DoF matches the upper bound.
\end{itemize}

\textit{Notation}: We use bold upper and lower case letters for matrices and column vectors, respectively. $\mathbb{C}^{m \times n}$ denotes the $m\times n$ dimensional complex space. $\mathbf{0}_{m\times n}$ and $\mathbf{I}_{n}$ represent the $m \times n$ zero matrix and the $n$-dimensional identity matrix, respectively. For any matrix $\mathbf{A}$, $\mathrm{vec}(\mathbf{A})$ denotes the vectorization of $\mathbf{A}$ formed by stacking the columns of $\mathbf{A}$ into a single column vector. Moreover, $\otimes$ represents the Kronecker product operation. 

\section{System Model}
\subsection{Channel Model}
Consider a symmetric multi-relay MIMO Y channel as shown in Fig. \ref{Fig:1}, where three user nodes, each equipped with $M$ antennas, exchange information with the help of $K$ relay nodes, each with $N$ antennas. Pairwise data exchange is employed, i.e., every user delivers two independent messages, one to each of the other two users. We assume that the information delivering is half-duplex, i.e., nodes in the network cannot transmit and receive signals simultaneously in a single frequency band. Every round of data exchange consists of two phases, namely, the uplink phase and the downlink phase. The two phases have equal duration $T$, where $T$ is an integer representing the number of symbols within each phase interval.

\begin{figure}
 \setlength{\abovecaptionskip}{-0.1cm}
  \centering
  \includegraphics[trim={0cm 2cm 0cm 0cm}, width=8cm]{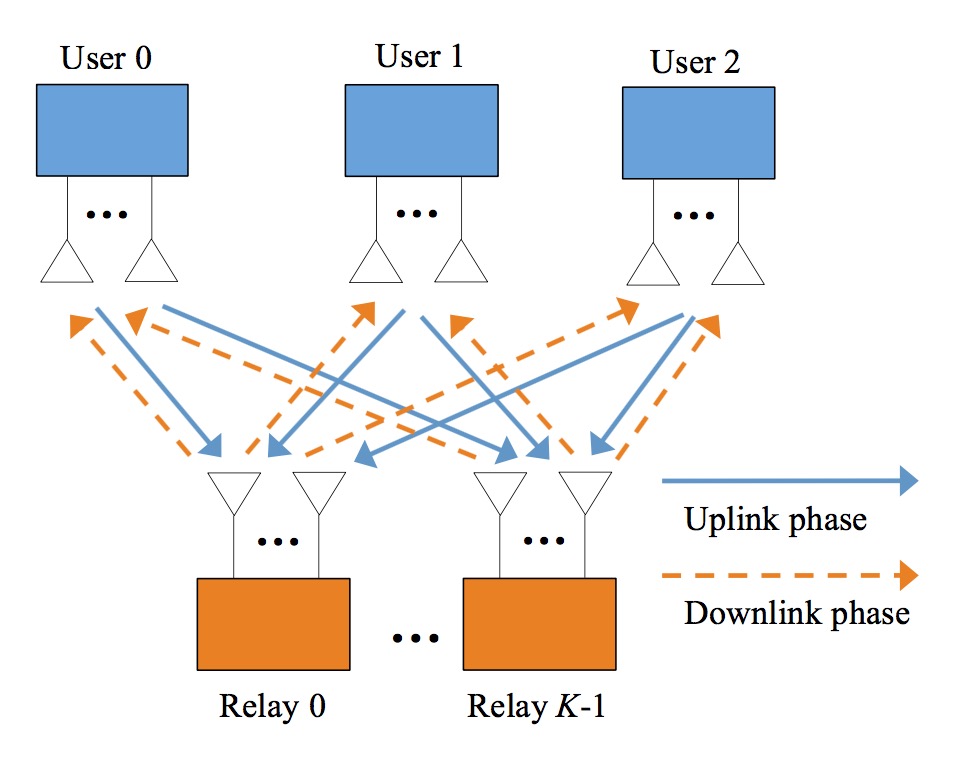}
  \caption{The system model of a symmetric multi-relay MIMO Y channel. The communication protocol consists of two phases: uplink phase and downlink phase.}\label{Fig:1}
\end{figure}

In the uplink phase, the users transmit signals to the relays simultaneously. The received signal at each relay is written by
\begin{equation}
\label{system1}
\mathbf{Y}_{\mathrm{R},k} =  \sum_{j=0}^2\mathbf{H}_{k,j}\mathbf{X}_j +\mathbf{Z}_{\mathrm{R},k}, \quad k= 0,1,\cdots,K-1
\end{equation}
where $\mathbf{H}_{k,j}\in \mathbb{C}^{N\times M}$ denotes the channel matrix from user $j$ to relay $k$; $\mathbf{X}_j \in \mathbb{C}^{M \times T}$ is the transmitted signal of user $j$; $\mathbf{Y}_{\mathrm{R},k} \in \mathbb{C}^{N \times T}$ is the received signal at relay $k$; $\mathbf{Z}_{\mathrm{R},k} \in \mathbb{C}^{N \times T}$ is the additive white Gaussian noise (AWGN) matrix at relay $k$, with the entries independently drawn from $\mathcal{CN}(0, \sigma_{\mathrm{R},k}^2)$. Note that $\sigma_{\mathrm{R},k}^2$ is the noise power at relay $k$. The power constraint for user $j$ is $\frac{1}{T}\mathrm{tr}(\mathbf{X}_{j}\mathbf{X}^H_{j}) \leq P_j$, where $P_j$ is the maximum transmission power allowed at user $j$.

In the downlink phase, the relays broadcast signals to the users. The received signal at each user is represented by
\begin{equation}
\label{system2}
\mathbf{Y}_j = \sum_{k=0}^{K-1}\mathbf{G}_{j,k}\mathbf{X}_{\mathrm{R},k} + \mathbf{Z}_j, \quad j= 0,1,2,
\end{equation}
where $\mathbf{G}_{j,k} \in \mathbb{C}^{M\times N}$ is the channel matrix from relay $k$ to user $j$; $\mathbf{X}_{\mathrm{R},k} \in \mathbb{C}^{N \times T}$ is the transmitted signal from relay $k$; $\mathbf{Y}_j \in \mathbb{C}^{M \times T}$ is the received signal at user $j$; $\mathbf{Z}_j \in \mathbb{C}^{M \times T}$ is the AWGN matrix at user $j$, with entries independently drawn from $\mathcal{CN}(0, \sigma_j^2)$. Here, $\sigma_{j}^2$ is the noise power at user $j$. The power constraint of relay $k$ is given by $\frac{1}{T}\mathrm{tr}\left(\mathbf{X}_{\mathrm{R},k}\mathbf{X}_{\mathrm{R},k}^H\right) \leq P_{\mathrm{R},k}$, where $P_{\mathrm{R},k}$ is the power budget of relay $k$.

The entries of channel matrices $\left\{\mathbf{H}_{k,j}\right\}$ and $\left\{\mathbf{G}_{j,k}\right\}$ are drawn from a continuous distribution, implying that the channel matrices are of full column or row rank, whichever is smaller, with probability one. We assume that channel state information (CSI) is globally known at every node in the model, following the convention in \cite{MIMO1,MIMO3,capacity1,Xchannel, Gao, multirelay, MIMO2, yuan}.\footnote{To realize the scheme in this paper, global CSI is sufficient but not necessary for every node. Each node only needs to know its linear processor designed in this scheme. There are many ways to achieve this. For example, we can employ a central controller that collects global CSI, computes the linear processors of the nodes, and then transmits the linear processors to their corresponding nodes. This will reduce to some extent the system overhead of global CSI acquisition at every node, without compromising the DoF.} Moreover, for notational convenience, we interpret the user index by modulo 3, e.g., user 3 is the same as user 0.

\subsection{Degrees of Freedom}
The goal of this paper is to analyze the degrees of freedom of the symmetric multi-relay MIMO Y channel described above. For convenience of discussion, we assume the same power constraint at each node, i.e., $P_0 = P_1 = P_2 = P$ and $P_{\mathrm{R},0} = P_{\mathrm{R},1} = \cdots = P_{\mathrm{R},K-1} = P$, which will not compromise the generality of the DoF results derived in this paper. Let $m_{j,j'} \in \{1,2,\cdots,2^{R_{j,j'}T}\}$ be the message from user $j$ to $j'$, where $R_{j,j'}$ is the corresponding information rate, for $j,j' =0,1,2$ and $j \not= j'$. Note that $R_{j,j'}$ is in general a function of power $P$, denoted by $R_{j,j'}(P)$. An information rate $R_{j,j'}(P)$ is said to be achievable if the error probability of decoding message $m_{j,j'}$ at receiver $j'$ approaches zero as $T \rightarrow \infty$. An achievable DoF of user $j$ to user $j'$ is defined as
\begin{equation}
d_{j,j'} = \lim\limits_{P \rightarrow \infty}\frac{R_{j,j'}(P)}{\log(P)}.
\end{equation}
Intuitively, $d_{j,j'}$ can be interpreted as the number of independent spatial data streams that user $j$ can reliably transmit to user $j'$ during each round of data exchange. An achievable total DoF of the symmetric multi-relay MIMO Y channel is defined as
\begin{equation}
\label{d_sum}
d_\mathrm{sum} = \frac{1}{2}\sum_{ \substack{ 0 \leq j,j' \leq 2\\j\not=j' \\}} d_{j,j'}.
\end{equation}
Note that the factor $\frac{1}{2}$ in \eqref{d_sum} is due to half-duplex communication. The optimal total DoF of the considered model, denoted by $d^\mathrm{opt}_\mathrm{sum}$, is defined as the supremum of $d_\mathrm{sum}$. In this paper, we assume a symmetric DoF setting with $d_{j,j'}=d$ for any $0\leq j,j' \leq 2$, $j \not =j'$. Then the total achievable DoF can be represented by $d_\mathrm{sum} = \frac{6d}{2} = 3d$.

We now present a DoF upper bound by assuming full cooperation among the relays. Under this assumption, the system model in \eqref{system1} and \eqref{system2} reduces to a single-relay MIMO Y channel, with $KN$ antennas at the relay. Therefore, the optimal DoF of such a single-relay MIMO Y channel naturally serves as a DoF upper bound of the model considered in \eqref{system1} and \eqref{system2}. From \cite{32}, this upper bound is given by
\begin{align}
\label{upperbound}
d_\mathrm{sum} \leq \min\left\{\frac{3M}{2}, KN\right\}.
\end{align}
Note that the optimal DoF in \cite{32} is derived for full-duplex communication. Thus the upper bound in \eqref{upperbound} is scaled by a factor of $\frac{1}{2}$ due to the half-duplex loss.

\subsection{Linear Processing}
In this paper, an achievable DoF of the symmetric multi-relay MIMO Y channel is derived by linear processing techniques. The message $m_{j,j'}$, $j' \not= j$, is encoded into $\mathbf{S}_{j,j'}\in \mathbb{C}^{d \times T}$, with $d$ independent spatial streams in $T$ channel uses. The transmitted signal of user $j$ is given by
\begin{equation}
\label{linear1}
\mathbf{X}_j = \sum_{j'\not=j}\mathbf{U}_{j,j'}\mathbf{S}_{j,j'},\quad j = 0,1,2,
\end{equation}
where $\mathbf{U}_{j,j'}\in \mathbb{C}^{M\times d}$ is the linear precoding matrix for $\mathbf{S}_{j,j'}$. An amplify-and-forward scheme is employed at the relays. Specifically, the transmitted signal of each relay is represented by
\begin{equation}
\label{linear2}
\mathbf{X}_{\mathrm{R},k} = \mathbf{F}_k\mathbf{Y}_{\mathrm{R},k},\quad  k=0,1,\cdots,K-1,
\end{equation}
where $\mathbf{F}_k \in \mathbb{C}^{N \times N}$ is the precoding matrix of relay $k$.

With \eqref{system1}, \eqref{linear1}, and \eqref{linear2}, we can express the received signal of user $j$ in \eqref{system2} as
\begin{align}
\label{received signal}
\mathbf{Y}_j = &\underbrace{\sum_{k = 0}^{K-1}\! \sum_{j' \not = j}\!\mathbf{G}_{j,k}\mathbf{F}_k\mathbf{H}_{k,j'}\mathbf{U}_{j',j}\mathbf{S}_{j',j}}_{\text{desired signal}} \! +\! \underbrace{\sum_{k = 0}^{K-1}\sum_{j' \not = j}\mathbf{G}_{j,k}\mathbf{F}_k\mathbf{H}_{k,j}\mathbf{U}_{j,j'}\mathbf{S}_{j,j'}}_{\text{self-interference}}\!+\!\underbrace{\sum_{k = 0}^{K-1}\!\!\sum_{j' \not = j'' \atop j', j'' \not =  j}\!\!\mathbf{G}_{j,k}\mathbf{F}_k\mathbf{H}_{k,j'}\mathbf{U}_{j',j''}\mathbf{S}_{j',j''}}_{\text{other interference}} \nonumber\\
& +\underbrace{\sum_{k = 0}^{K-1}\mathbf{G}_{j,k}\mathbf{F}_k\mathbf{Z}_{\mathrm{R},k}+\mathbf{Z}_j}_{\text{noise}}.
\end{align}
In the above, $\mathbf{Y}_j$ consists of four signal components: the desired signal, the self interference, the other interference and the noise. Since user $j$ perfectly knows the CSI and the self message $\left\{\mathbf{S}_{j,j'}, \forall j' \not= j\right\}$, the self-interference term in \eqref{received signal} can be pre-cancelled before further processing. Each user $j$ is required to decode $2d$ spatial streams, $d$ from each of the other two users. To this end, there must be an interference-free subspace with dimension $2d$ in the receiving signal space of user $j$. More specifically, denote by $\mathbf{V}_j \in \mathbb{C}^{2d \times M}$ a projection matrix, with $\mathbf{V}_j\mathbf{Y}_j$ being the projected image of $\mathbf{Y}_j$ in the subspace spanned by the row space of $\mathbf{V}_j$. Then, to ensure the decodability of $\mathbf{S}_{j',j}$ at user $j$, we should appropriately design $\{\mathbf{U}_{j,j'}\}$, $\{\mathbf{F}_k\}$, and $\{\mathbf{V}_j\}$ to satisfy two sets of requirements, as detailed below.

First, $\mathbf{V}_j\mathbf{Y}_j$ should be free of interference. That is, the following interference neutralization requirements should be met:
\begin{subequations}
\label{zeroforcing}
\begin{align}
\label{zeroforcing_1}
\sum_{k=0}^{K-1}\mathbf{V}_0\mathbf{G}_{0,k}\mathbf{F}_k\mathbf{H}_{k,1}\mathbf{U}_{1,2}=\mathbf{0},\quad &\quad 
\sum_{k=0}^{K-1}\mathbf{V}_0\mathbf{G}_{0,k}\mathbf{F}_k\mathbf{H}_{k,2}\mathbf{U}_{2,1}=\mathbf{0},\\
\label{zeroforcing_12}
\sum_{k=0}^{K-1}\mathbf{V}_1\mathbf{G}_{1,k}\mathbf{F}_k\mathbf{H}_{k,2}\mathbf{U}_{2,0}=\mathbf{0},\quad &\quad 
\sum_{k=0}^{K-1}\mathbf{V}_1\mathbf{G}_{1,k}\mathbf{F}_k\mathbf{H}_{k,0}\mathbf{U}_{0,2}=\mathbf{0},\\
\label{zeroforcing_2}
\sum_{k=0}^{K-1}\mathbf{V}_2\mathbf{G}_{2,k}\mathbf{F}_k\mathbf{H}_{k,0}\mathbf{U}_{0,1}=\mathbf{0},\quad &\quad 
\sum_{k=0}^{K-1}\mathbf{V}_2\mathbf{G}_{2,k}\mathbf{F}_k\mathbf{H}_{k,1}\mathbf{U}_{1,0}=\mathbf{0}.
\end{align}
Here, ``interference neutralization" refers to a special transceiver design strategy for interference cancellation such that a common source of interference from different paths cancels itself at a destination. Second, user $j$ needs to decode $2d$ spatial streams from the projected signal $\mathbf{V}_j\mathbf{Y}_j \in \mathbb{C}^{2d\times T}$. To ensure the decodability, the desired signal in \eqref{received signal} after projection should be of rank $2d$. Define
\begin{equation}
\mathbf{W}_{k,j} = \left[\mathbf{H}_{k,j+1}\mathbf{U}_{j+1,j},\mathbf{H}_{k,j-1}\mathbf{U}_{j-1,j}\right], \quad\!\!j = 0,1,2. \nonumber
\end{equation}
Then $\sum_{k=0}^{K-1}\mathbf{V}_j\mathbf{G}_{j,k}\mathbf{F}_k\mathbf{W}_{k,j}$ represents the effective channel for the messages desired by user $j$. To ensure the decodability of $2d$ spatial streams at each user $j$, we have the following rank requirements:
\begin{align}
\label{rank}
\mathrm{rank}\left(\sum_{k=0}^{K-1}\mathbf{V}_j\mathbf{G}_{j,k}\mathbf{F}_k\mathbf{W}_{k,j}\right)=2d, \quad\!\! j=0,1,2.
\end{align}
\end{subequations}

Given an antenna setup $\left(M,N\right)$ and a target DoF $d$, if there exist suitable $\{\mathbf{U}_{j',j},\mathbf{F}_k,\mathbf{V}_j\}$ satisfying \eqref{zeroforcing} for randomly generated channel matrices $\{\mathbf{H}_{k,j},\mathbf{G}_{j,k}\}$ with probability one, then a total DoF $d_\mathrm{sum}=3d$ is achieved by the proposed linear processing scheme. Thus, the key issue is to analyze the solvability of the system \eqref{zeroforcing} with respect to $\{\mathbf{U}_{j',j},\mathbf{F}_k,\mathbf{V}_j\}$, which is the main focus of the rest of this paper.

\section{Achievable DoF of the Symmetric Multi-Relay MIMO Y Channel}
In general, to check the achievability of a certain DoF $d$, we need to jointly design the matrices $\{\mathbf{U}_{j,j'},\mathbf{F}_k,\mathbf{V}_j\}$ to meet \eqref{zeroforcing}. This is a challenging task since the equations in \eqref{zeroforcing} are nonlinear with respect to $\{\mathbf{U}_{j,j'},\mathbf{F}_k,\mathbf{V}_j\}$. To tackle this problem, we start with a conventional approach based on the idea of uplink-downlink symmetry.

\subsection{Conventional Approach with Uplink-Downlink Symmetry}
Uplink-downlink symmetry has been widely used in precoding design for MIMO mRCs \cite{capacity1,yuan,3,32,capacity2}. It is shown to be optimal for many single-relay MIMO mRCs \cite{capacity1,3, 32}, and efficient for some multi-relay MIMO mRCs \cite{capacity2}. In this subsection, we follow the idea of uplink-downlink symmetry to solve \eqref{zeroforcing}. Then, for an arbitrary configuration of $(M,N,K)$, we derive an achievable DoF (or an upper bound of the achievable DoF) of the symmetric multi-relay MIMO Y channel. We show that there is a significant DoF gap between this result and the DoF upper bound in \eqref{upperbound}, implying the inadequacy of the uplink-downlink symmetric precoding design for multi-relay MIMO mRCs.

To start with, we split each projection matrix $\mathbf{V}_j$ equally into two parts as
\begin{equation}
\mathbf{V}_j = \left[\mathbf{V}_{j+1,j}, \mathbf{V}_{j-1,j}\right]
\end{equation}
where $\mathbf{V}_{j',j}\in \mathbb{C}^{d\times M}$ is the projection matrix for the message $\mathbf{S}_{j',j}$. Then the interference neutralization conditions \eqref{zeroforcing_1}-\eqref{zeroforcing_2} can be rewritten as
\begin{equation}
\label{zeroforcing_symg}
\sum_{k=0}^{K-1}\mathbf{V}_{j',j}\mathbf{G}_{j,k}\mathbf{F}_k\mathbf{H}_{k,j'}\mathbf{U}_{j',j''}=\mathbf{0},\quad\!\! \forall j \not= j',\quad\!\!\! j'\not=j'',\quad\!\!\!j \not = j''.
\end{equation}
Note that for all the 3 users, there are 12 matrix equations in \eqref{zeroforcing_symg} for interference neutralization. Moreover, the rank requirements remain the same as \eqref{rank}.

We next establish $K$ DoF points, namely, $\left(\frac{M}{N}, d_\mathrm{sum}\right) = \left(\frac{2}{3},N\right)$ and $\left(\frac{6k+\sqrt{6k}}{12},\frac{\sqrt{6k}N}{2}\right)$, $k = 2,\cdots, K$, by following uplink-downlink symmetric design. The DoF point $\left(\frac{2}{3},N\right)$ is achievable, while the points $\left(\frac{6k+\sqrt{6k}}{12},\frac{\sqrt{6k}N}{2}\right)$, $k = 2,\cdots, K$, are just upper bounds. That is, for any $\frac{M}{N} \geq \frac{6k+\sqrt{6k}}{12}$, $k=2,\cdots,K$, the total DoF achieved by the uplink-downlink symmetric design is upper bounded as $d_\mathrm{sum} < \frac{\sqrt{6k}N}{2}$.

The first DoF point is derived by deactivating $K-1$ of the $K$ relays. In this case, the model reduces to a single-relay MIMO Y channel with the antenna number of the relay equal to $N$. Then, the precoding design in \cite{32} can be applied directly. From the result of \cite{32}, a total DoF of $\frac{3M}{2}$ can be achieved for half-duplex single-relay MIMO Y channel with $\frac{M}{N} = \frac{2}{3}$. That is, $\left(\frac{M}{N}, d_\mathrm{sum}\right)=\left(\frac{2}{3},N\right)$ is achievable. 

We now consider the DoF point $\left(\frac{6K+\sqrt{6K}}{12},\frac{\sqrt{6K}N}{2}\right)$. Note that the remaining DoF points can be straightforwardly obtained by deactivating $K-k$ relays, for $k=2,\cdots,K-1$. Following \cite{capacity2}, we apply signal alignment techniques for the design of $\{\mathbf{U}_{j,j'},\mathbf{V}_{j,j'}\}$ to reduce the number of linearly independent constraints in \eqref{zeroforcing_sym}. First consider the uplink signal alignment design. We align the signals exchanged by user $j$ and user $j'$ at each relay. That is, we design $\{\mathbf{U}_{j,j'}\}$ to satisfy
\begin{align}
\label{uplinkalign}
\mathbf{H}_{k,j}\mathbf{U}_{j,j'} = \mathbf{H}_{k,j'}\mathbf{U}_{j',j}, \quad j,j' = 0,1,2, \quad j\not =j', \quad \forall k.
\end{align}
Note that for the single-relay case, it usually suffices to align $\mathbf{H}_{k,j}\mathbf{U}_{j,j'}$ and $\mathbf{H}_{k,j'}\mathbf{U}_{j',j}$ in a common subspace. However, for the multi-relay case here, we rely on a more strict constraint \eqref{uplinkalign} for signal alignment, so as to reduce the number of linearly independent equations in \eqref{zeroforcing_sym}. We then consider the downlink signal alignment. From \eqref{zeroforcing_symg}, we see that the uplink equivalent channel matrix $\mathbf{H}_{k,j}\mathbf{U}_{j,j'}\in\mathbb{C}^{N\times d}$ is of the same size as the transpose of downlink equivalent channel matrix $\mathbf{V}_{j,j'}\mathbf{G}_{j',k} \in \mathbb{C}^{d \times N}$, for $\forall k,j,j'$. This structural symmetry implies that any beamforming design in the uplink phase directly carries over to the downlink phase. Specifically, we design the downlink receiving matrices $\{\mathbf{V}_{j,j'}\}$ to satisfy
\begin{align}
\label{downlinkalign}
\mathbf{V}_{j,j'}\mathbf{G}_{j',k} = \mathbf{V}_{j',j}\mathbf{G}_{j,k}, \quad j,j' = 0,1,2, \quad j\not =j', \quad \forall k.
\end{align}
From the rank-nullity theorem, to ensure the existence of full-rank $\{\mathbf{U}_{j,j'},\mathbf{V}_{j,j'}\}$ satisfying \eqref{uplinkalign} and \eqref{downlinkalign}, the following conditions must be met:
\begin{align}
\label{symmetry_align_condition}
2M-KN\geq d.
\end{align}
With \eqref{uplinkalign} and \eqref{downlinkalign}, the interference neutralization conditions in \eqref{zeroforcing_symg} reduces to 
\begin{equation}
\label{zeroforcing_sym}
\sum_{k=0}^{K-1}\mathbf{V}_{j',j}\mathbf{G}_{j,k}\mathbf{F}_k\mathbf{H}_{k,j+1}\mathbf{U}_{j+1,j+2}=\mathbf{0},\quad \!\! \quad j,j' = 0,1,2, \quad j\not =j'.
\end{equation}
Note that $\{\mathbf{U}_{j',j},\mathbf{V}_{j',j}\}$ are already fixed to meet \eqref{uplinkalign} and \eqref{downlinkalign}. Thus \eqref{zeroforcing_sym} is a linear system of $\{\mathbf{F}_k\}$ with $6d^2$ equations and $KN^2$ unknown variables. The system has a non-zero solution of $\{\mathbf{F}_{k}\}$ provided $d < \frac{\sqrt{6K}N}{6}$. Together with \eqref{symmetry_align_condition}, we obtain the DoF point $\left(\frac{M}{N}, d_\mathrm{sum}\right) = \left(\frac{6K+\sqrt{6K}}{12},\frac{\sqrt{6K}N}{2}\right)$. Note that by deactivating $K-k$ relays, we immediately obtain the remaining DoF points $\left(\frac{M}{N}, d_\mathrm{sum}\right) = \left(\frac{6k+\sqrt{6k}}{12},\frac{\sqrt{6k}N}{2}\right)$, for $k = 2,\cdots, K-1$.

We now apply the antenna disablement lemma \cite{DOF2} to the above $K$ DoF points, yielding a continuous DoF curve of uplink-downlink symmetric design for an arbitrary value of $\frac{M}{N}$:
\begin{equation}
\label{DoF_sym}
d_\mathrm{sum} = N\max_{(a,b)\in \mathcal{S}_K}g_{(a,b)}\left(\frac{M}{N}\right)
\end{equation}
where $\mathcal{S}_K = \left\{\left(\frac{2}{3},1\right)\right\} \cup \left\{\left( \frac{6k+\sqrt{6k}}{12},\frac{\sqrt{6k}}{2}\right) \Big| k = 2,\cdots, K\right\}$ and the $g$-function is defined as
\begin{equation}
g_{(a,b)}(x) = \begin{cases}
\frac{bx}{a}  &       x<a \\
b & x\geq a\\
\end{cases}.
\end{equation}

It is worth noting that the DoF in \eqref{DoF_sym} is not necessarily achievable. To verify its achievability, we still need to check whether the solution of $\{\mathbf{F}_{k}\}$ to \eqref{zeroforcing_sym} satisfies the rank requirements \eqref{rank}. In general, the DoF in \eqref{DoF_sym} serves as an upper bound for the DoF achieved by the uplink-downlink symmetric design described in this subsection. In fact, this bound is enough to illustrate the limitation of the uplink-downlink symmetric design. Specifically, we compare \eqref{DoF_sym} with the upper bound in \eqref{upperbound}. We see that for a sufficiently large $\frac{M}{N}$ (say, $\frac{M}{N} \geq \frac{6K+\sqrt{6K}}{12}$), $d_\mathrm{sum}$ in \eqref{DoF_sym} is proportional to $\sqrt{K}$ while the DoF upper-bound in
\eqref{upperbound} and the achievable DoF derived in this paper (Theorem \ref{Theorem1} in Section \ref{Main Result}) both increase linearly with $K$. This DoF gap implies the inefficiency of the uplink-downlink symmetric design.

\subsection{Main Result}
\label{Main Result}
We now propose a new and more efficient approach to the precoding design for the symmetric multi-relay MIMO Y channel, so as to achieve a higher DoF than the result in \eqref{DoF_sym}. The novelty of our approach is as follows. First, instead of following the conventional uplink-downlink symmetry, our approach performs signal alignment only in the uplink phase. This means that, in the design of user precoding and post-processing, only $\{\mathbf{U}_{j,j'}\}$ are designed to align signals to reduce the number of linearly independent equations in \eqref{zeroforcing_1}-\eqref{zeroforcing_2}; the post-processors $\{\mathbf{V}_j\}$ are used directly for the purpose of neutralizing interference, rather than designed by following the uplink-downlink symmetry. This introduces additional freedom for system design. Further, to exploit the advantage of this uplink-downlink asymmetric design, we allow a certain number of \textit{receiving} (but not transmitting) antennas at each relay to be deactivated. With this treatment, signal alignment can be performed in a larger range of $\frac{M}{N}$. With the above design of $\{\mathbf{U}_{j,j'}\}$ and $\{\mathbf{V}_{j}\}$, the problem reduces to the design of the relay precoders $\{\mathbf{F}_k\}$ that satisfy the linear system in \eqref{zeroforcing_1}-\eqref{zeroforcing_2} and the rank requirements in \eqref{rank}. We then establish a general method to address this solvability problem. We emphasis that this method can be used to analyze the solvability of any linear system with rank constraints, and so, will find applications to the DoF analysis of many other multiway relay networks.

The main result of this paper is presented below, with the proof given in Section IV.

\newtheorem{theorem}{Theorem}
\newtheorem{lemma}{Lemma}
\newtheorem{corollary}{Corollary}
\newtheorem{remark}{Remark}
\begin{theorem}
\label{Theorem1}
For the symmetric multi-relay MIMO Y channel, any total DoF $d_\mathrm{sum}<d^*_\mathrm{sum}$ is achievable, where
\begin{equation}
\label{achievableDoF0}
d^*_\mathrm{sum} = \mathrm{min}\left\{\frac{3M}{2}, \max\left\{M + \frac{5MN}{9M+N},\frac{\sqrt{3K}N}{2}\right\}, M + \frac{KN^2}{3M}, KN\right\}.
\end{equation}
\end{theorem}

\begin{corollary}
\label{Corollary2}
For $K=1$, any total DoF $d_\mathrm{sum}<d^*_\mathrm{sum}$ is achievable, where
\begin{equation}
\label{MIMOYcapacity}
d^*_\mathrm{sum} = \mathrm{min}\left\{\frac{3M}{2}, N\right\}.
\end{equation}
\end{corollary}

\begin{proof}
By setting $K = 1$, Corollary 1 follows directly from Theorem 1.
\end{proof}

\begin{remark}
\normalfont
For $K=1$, the considered model reduces to signal-relay MIMO Y channel. The achievable DoF in Corollary \ref{Corollary2} is exactly the optimal DoF derived in \cite{32}. Hence, our result in Theorem \ref{Theorem1} subsumes the result in \cite{32} and generalizes it for the case with multiple relays.
\end{remark}

For $K\geq 2$, the result in Theorem 1 is given piecewisely in the following corollaries.

\begin{corollary}
For $K \geq 2$, any total DoF $d_\mathrm{sum}<d^*_\mathrm{sum}$ is achievable, where
\begin{equation}
\label{achievableDoF}
d^*_\mathrm{sum} =
\begin{cases}
\frac{3M}{2}  &        {\frac{M}{N} \in \Big[0,\max\left\{\frac{\sqrt{3K}}{3},1\right\}\Big)}\\
\max\left\{M + \frac{5MN}{9M+N},\frac{\sqrt{3K}N}{2}\right\}      & {\frac{M}{N} \in \Big[\max\left\{\frac{\sqrt{3K}}{3},1\right\}, \frac{9K+\sqrt{81K^2+60K}}{30}\Big)}  \\
M + \frac{KN^2}{3M}     & {\frac{M}{N} \in \Big[\frac{9K+\sqrt{81K^2+60K}}{30},\frac{3K+\sqrt{9K^2-12K}}{6}\Big)} \\
KN  &      {\frac{M}{N} \in \Big[\frac{3K+\sqrt{9K^2-12K}}{6},\infty\Big)}.
\end{cases}
\end{equation}
\end{corollary}
\begin{proof}
By setting $K\geq2$, the result in \eqref{achievableDoF} follows from Theorem 1.
\end{proof}

\begin{remark}
\normalfont
We see that for $K \geq 2$, the achievable total DoF in Theorem \ref{Theorem1} is given by the minimum of four terms: $\frac{3M}{2}$, $\max\left\{M + \frac{5MN}{9M+N},\frac{\sqrt{3K}N}{2}\right\}$, $M + \frac{KN^2}{3M}$, and $KN$. Note that the DoF $M + \frac{5MN}{9M+N}$ and $M + \frac{KN^2}{3M}$ are achieved by two different signal alignment approaches, while the DoF $\frac{\sqrt{3K}N}{2}$ is achieved without signal alignment. The achievability of the DoF $\frac{3M}{2}$ and $KN$ is derived by the technique of antenna disablement.
\end{remark}

\begin{corollary}
For $K\geq 2$ and $\frac{M}{N} \in \Big[0,\max\left\{\frac{\sqrt{3K}}{3},1\right\} \Big)\cup \Big[\frac{3K+\sqrt{9K^2-12K}}{6},\infty\Big)$, the optimal total DoF of the symmetric multi-relay MIMO Y channel is given by
\begin{equation}
\label{optimalDoF}
d^\mathrm{opt}_\mathrm{sum} =
\mathrm{min}\left\{\frac{3M}{2},KN\right\}.
\end{equation}
\end{corollary}
\begin{proof}
For $\frac{M}{N} \in \Big[0,\max\left\{\frac{\sqrt{3K}}{3},1\right\} \Big) \cup \Big[\frac{3K+\sqrt{9K^2-12K}}{6},\infty\Big)$, the achievable total DoF in Corollary 2 coincides with the upper bound \eqref{upperbound}, and therefore the optimal total DoF is achieved.
\end{proof}

\begin{remark} 
\normalfont For $\frac{M}{N} \in \Big(\max\left\{\frac{\sqrt{3K}}{3},1\right\}, \frac{3K+\sqrt{9K^2-12K}}{6}\Big)$ (with $K\geq 2$), our achievable total DoF does not match the upper bound in \eqref{upperbound}. We conjecture that the upper bound obtained by assuming full cooperation among the relays is loose in general. Tighter DoF upper bounds can be derived by characterizing the DoF degradation due to distributed relaying. This will be an interesting topic for future research. 
\end{remark}

Fig. \ref{Fig:21} and Fig. \ref{Fig:22} illustrate the achievable total DoF in Corollary 2 with $K=2$ and $K=5$, respectively, together with the DoF in \eqref{DoF_sym} and the DoF upper bound in \eqref{upperbound}. We see that our approach significantly outperforms the conventional approach based on the uplink-downlink symmetric design. Particularly, as $K \rightarrow \infty$, the corresponding gap increases unboundedly. For $\frac{M}{N} \in \left(0,\max\left\{\frac{\sqrt{3K}}{3},1\right\}\right) \cup \left(\frac{3K+\sqrt{9K^2-12K}}{6}, \infty \right)$, our achievable total DoF coincides with the DoF upper bound given in \eqref{upperbound}, indicating that the optimal total DoF of the considered model is achieved. However, for the uplink-downlink symmetric design, the optimal total DoF is achieved only for $\frac{M}{N} \in \left(0,\frac{2}{3}\right]$. We also see that for sufficiently large $\frac{M}{N}$, there is a constant gap between the achievable total DoF of conventional approach and that of our approach, implying that the uplink-downlink symmetric approach can not make full use of all the relay antennas.

\begin{figure}
\setlength{\abovecaptionskip}{-0.1cm}
  \centering
  \includegraphics[trim={0 1cm 0 0}, width=10cm]{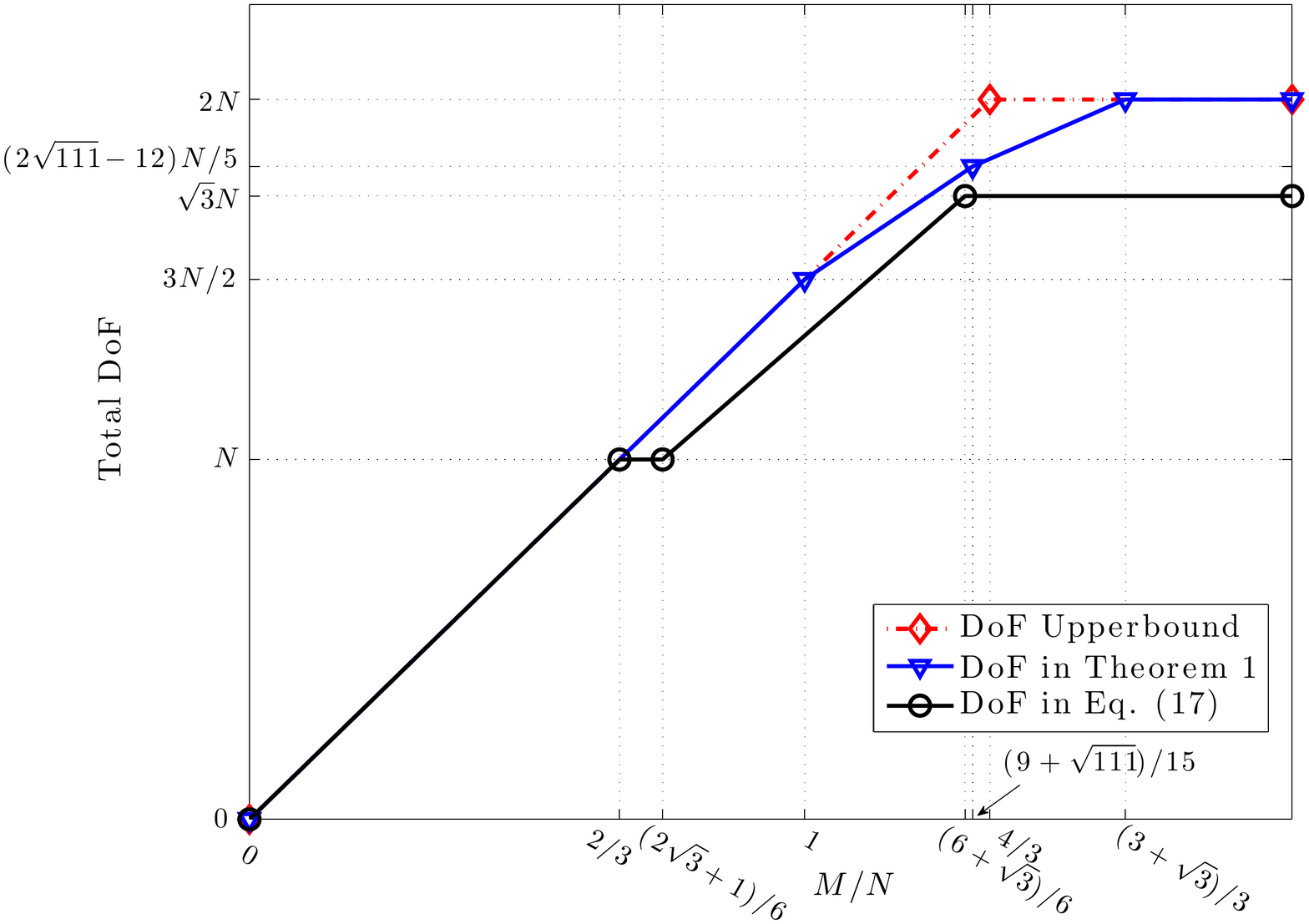}
  \caption{Achievable total DoF of the symmetric multi-relay MIMO Y channel for $K=2$.}\label{Fig:21}
\end{figure}

We now study the asymptotic behavior of the achievable total DoF as $K$ goes to infinity. It is readily seen from Corollary 2 that, for large $\frac{M}{N}$ and any fixed $N$, the achievable total DoF increases linearly in $K$. What is more interesting is to understand the asymptotic DoF behavior by studying the achievable total DoF normalized by the total number of relay antennas, i.e., $\frac{d_\mathrm{sum}}{N_\mathrm{total}}$, where $N_\mathrm{total} = KN$. This limiting process illustrates how the degree of cooperation among relay antennas affects the DoF of the system. We have the following asymptotic bound.

\begin{corollary}
For the symmetric multi-relay MIMO Y channel, as $K$ and $M$ go to infinity at the same rate, the asymptotic bound of $\frac{d_\mathrm{sum}}{N_\mathrm{total}}$ is given by 
\begin{equation}
\label{asymptotic}
\frac{d^*_\mathrm{sum}}{N_\mathrm{total}} = \mathrm{min}\left\{\frac{M}{N_\mathrm{total}} , 1\right\}.
\end{equation}
\end{corollary}
\begin{proof}
By dividing both sides of \eqref{achievableDoF} by $N_\mathrm{total} = KN$, we obtain $\frac{d^*_\mathrm{sum}}{N_\mathrm{total}}$ with respect to $\frac{M}{N_\mathrm{total}}$:
\begin{equation}
\label{asymptoticproof}
\frac{d^*_\mathrm{sum}}{N_\mathrm{total}} =
\begin{cases}
\frac{3M}{2N_\mathrm{total}}  &        {\frac{M}{N_\mathrm{total}} \in \Big[0,\max \left\{\frac{\sqrt{3K}}{3K},\frac{1}{K}\right\}\Big)}\\
\max\left\{\frac{M}{N_\mathrm{total}} + \frac{5M}{K(9M+N)},\frac{\sqrt{3K}}{2K}\right\}      & {\frac{M}{N_\mathrm{total}} \in \Big[\max\left\{\frac{\sqrt{3K}}{3K},\frac{1}{K}\right\},\frac{9K+\sqrt{81K^2+60K}}{30K}\Big)}\\
\frac{M}{N_\mathrm{total}} + \frac{N}{3M}     & {\frac{M}{N_\mathrm{total}} \in \Big[\frac{9K+\sqrt{81K^2+60K}}{30K},\frac{3K+\sqrt{9K^2-12K}}{6K}\Big)} \\
1  &      {\frac{M}{N_\mathrm{total}} \in \Big[\frac{3K+\sqrt{9K^2-12K}}{6},\infty\Big)}.
\end{cases}
\end{equation}
Letting $K\rightarrow \infty$ and $M\rightarrow \infty$ at the same rate, we obtain \eqref{asymptotic}.
\end{proof}

\begin{figure}
\setlength{\abovecaptionskip}{-0.1cm}
 \centering
  \includegraphics[trim={0 1cm 0 0},width=10cm]{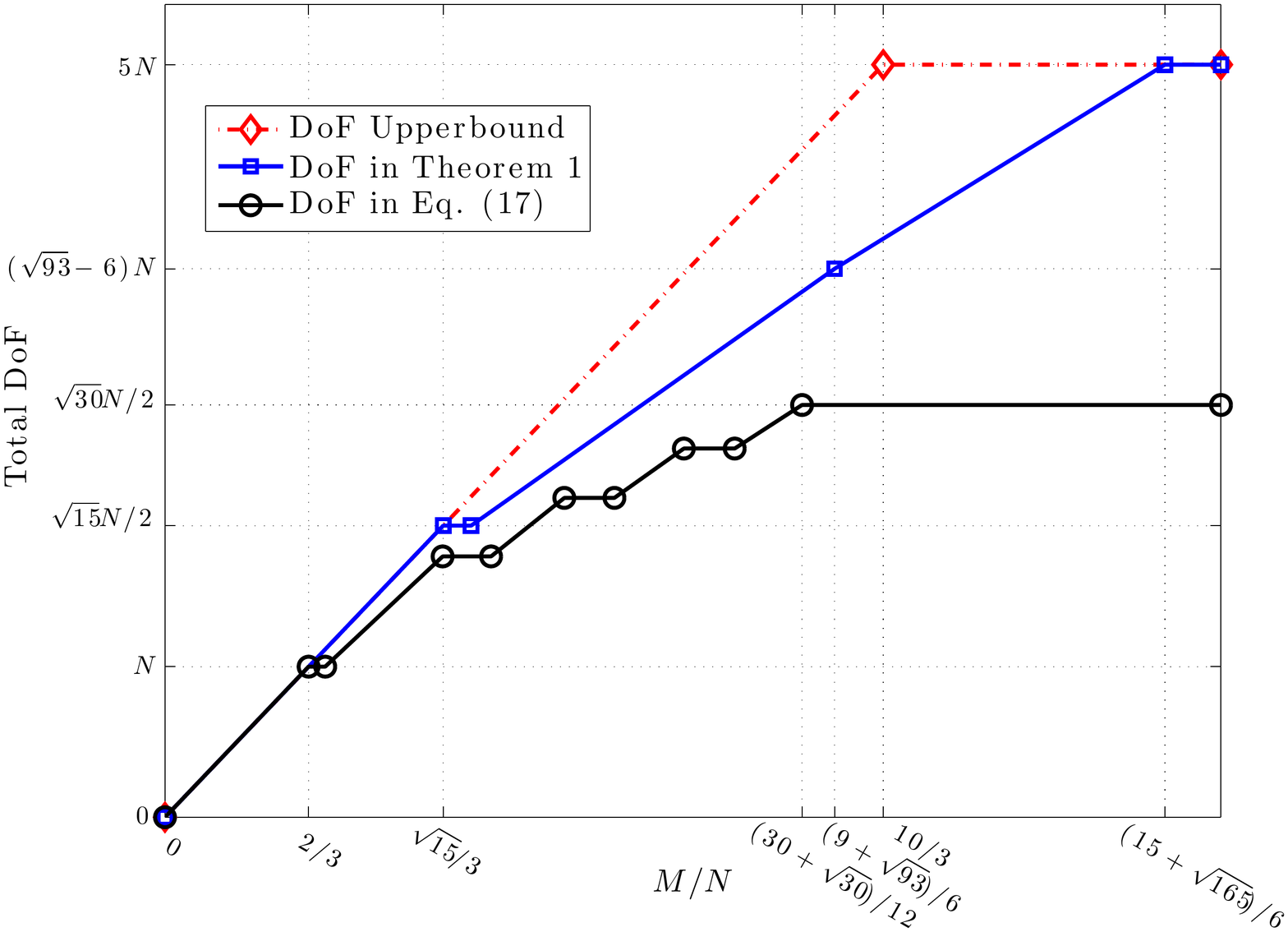}
  \caption{Achievable total DoF of the symmetric multi-relay MIMO Y channel for $K=5$.}\label{Fig:22}
\end{figure}

In Fig. \ref{Fig:3}, we compare the achievable total DoF (normalized by $N_\mathrm{total}$) with respect to $\frac{M}{N_\mathrm{total}}$ for different values of $K$. Our target is to see the impact on the achievable total DoF when a single $KN$-antenna relay is split into $K$ relay nodes. Fig. \ref{Fig:3} illustrates the normalized achievable DoF for $K=1,2,5$ and $10$. Note that the case of $K=1$ assumes that all the $KN$ antennas at the relay can cooperate, and so serves as an upper bound of the normalized achievable DoF. From Fig. \ref{Fig:3}, we see that as $K$ increases, for a given $\frac{M}{N_\mathrm{total}}$, the normalized achievable DoF monotonically decreases and approaches the asymptotic bound in Corollary 4. 

Fig. \ref{Fig:4} plots the achievable total DoF versus the number of relays for various values of $\frac{M}{N}$. Clearly, $d_\mathrm{sum}$ cannot exceed $\frac{3M}{2}$, the case in which all the user antennas are fully utilized. So, for a fixed value of $\frac{M}{N}$, it is interesting to explore how the achievable DoF increases with respect to $K$ and the number of relay nodes required to achieve $d_\mathrm{sum} = \frac{3M}{2}$. From Fig. \ref{Fig:4}, we see that the achievable total DoF is roughly a piecewise linear function with respect to $K$. The minimum numbers of relays required to achieve $d_\mathrm{sum} = \frac{3M}{2}$ are $2,12,19,27$ for $\frac{M}{N}=1,2,\frac{5}{2},3$, respectively. Fig. \ref{Fig:4} also shows that, if a certain amount of inefficiency can be tolerated in utilizing user antennas, then the required number of relays can be considerably reduced. For example, only 5 relays are required to achieve $d_\mathrm{sum}=\frac{7N}{2} = \frac{7}{6}M$ for $\frac{M}{N}=3$, while the number is 27 in achieving $d_\mathrm{sum}=\frac{3M}{2}$. It is also worth noting that there are transitional flat DoF curves for the cases of $\frac{M}{N}=2,\frac{5}{2}$, and 3. This flatness corresponds to the result $d_\mathrm{sum}^*=M+\frac{5MN}{9M+N}$ (that is independent of $K$).

\begin{figure}
  \centering
  \setlength{\abovecaptionskip}{-0.1cm}
  \includegraphics[trim={0 1cm 0 1cm},width=10cm]{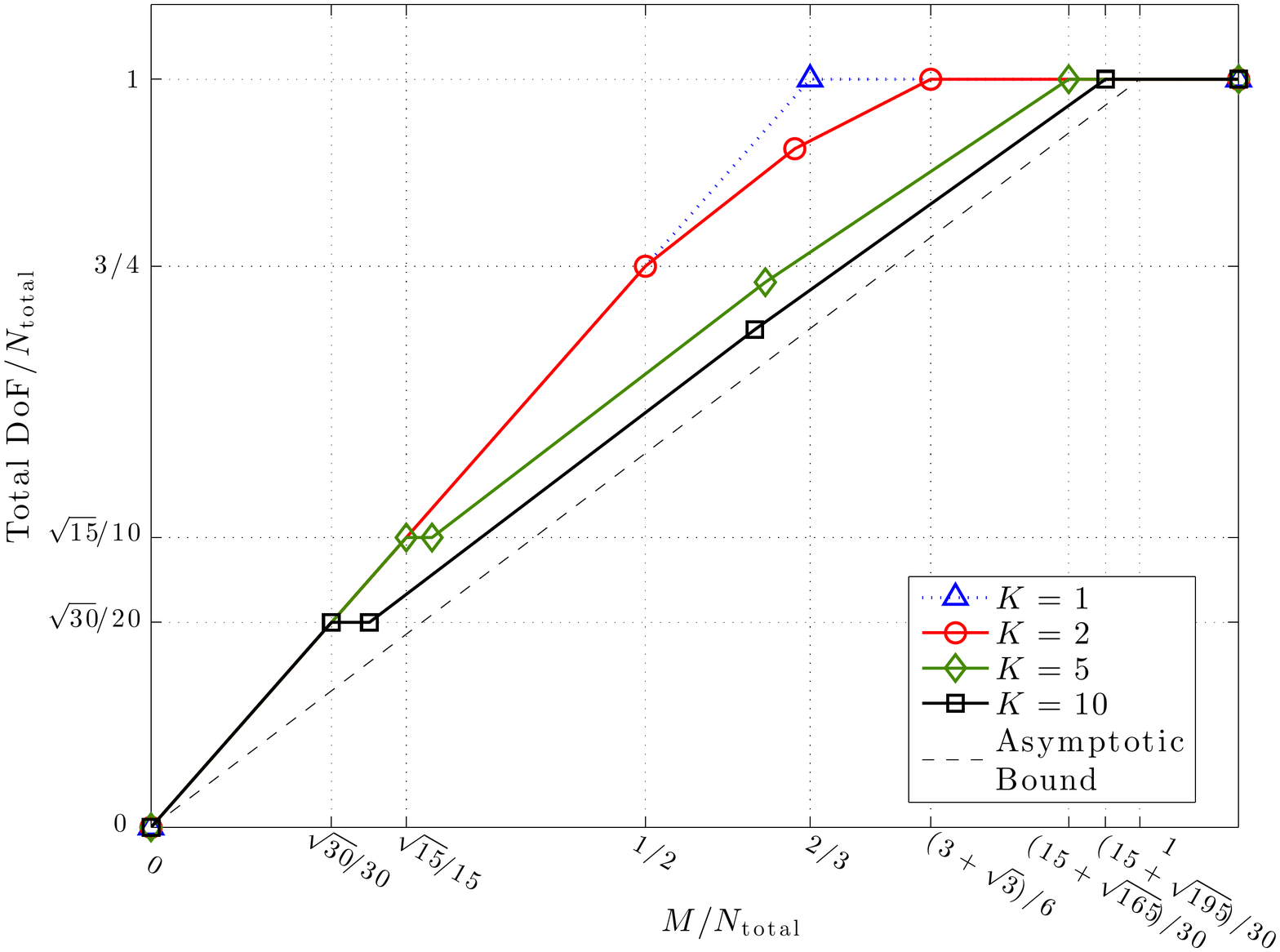}
  \caption{The achievable total DoF against $\frac{M}{N_\mathrm{total}}$, for $K=2,5,$ and $10$.}\label{Fig:3}
\end{figure}

\begin{figure}
  \centering
  \setlength{\abovecaptionskip}{-0.1cm}
  \includegraphics[trim={0 1cm 0 0},width=10cm]{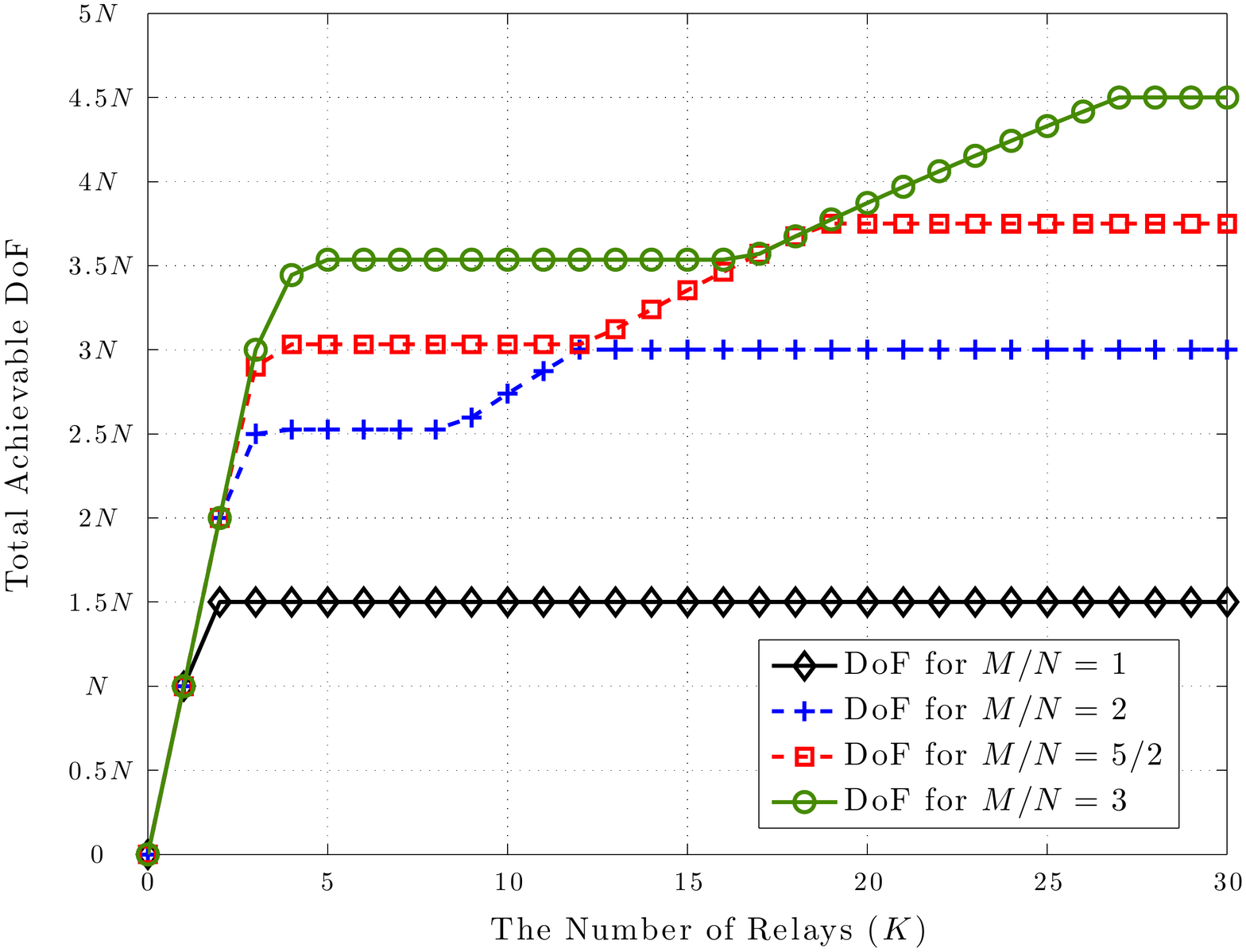}
  \caption{The achievable total DoF against the number of relays.}\label{Fig:4}
\end{figure}

\section{Proof of Theorem 1}
\label{Proof of Theorem 1}
This section is dedicated to the proof of Theorem 1. For $K=1$, Theorem 1 has been proven in \cite{3}. Thus, it suffices to only consider the case of $K\geq 2$. Based on the linear processing techniques in Section II-B, our goal is to jointly design $\{\mathbf{U}_{j,j'}, \mathbf{F}_k,\mathbf{V}_j\}$ to satisfy condition \eqref{zeroforcing}. We propose two types of signal alignment methods for $\{\mathbf{U}_{j,j'}\}$ to generate redundant equations in \eqref{zeroforcing_1}-\eqref{zeroforcing_2}, as detailed below.

\subsection{Signal Alignment I}
Signal Alignment I is described as follows. For $j = 0,1,2$, $\mathbf{U}_{j,j+1}$ and $\mathbf{U}_{j+1,j}$ are aligned to satisfy
\begin{align}
\label{pairwisealignment}
\mathbf{H}_{k,j}\mathbf{U}_{j,j+1} =  \mathbf{H}_{k,j+1}\mathbf{U}_{j+1,j}, \quad \!\!k=0,1,\cdots,K-1,
\end{align}
or equivalently
\begin{subequations}
\begin{align}
\label{si}
\mathbf{K}_j \underbrace{\left[\mathbf{U}_{j,j+1}^T,\mathbf{U}^T_{j+1,j}\right]^T}_{2M\times d}  = \mathbf{0}
\end{align}
where
\begin{align}
\label{K_j}
\mathbf{K}_j = 
\left[\begin{array}{cc}
    \mathbf{H}_{0,j} & - \mathbf{H}_{0,j+1}\\
    \vdots & \vdots\\
    \mathbf{H}_{K-1,j} & - \mathbf{H}_{K-1,j+1}
  \end{array}\right] \in \mathbb{C}^{KN \times 2M}.
\end{align}
\end{subequations}
From the rank-nullity theorem, there exist full-rank $\mathbf{U}_{j,j+1}$ and $\mathbf{U}_{j+1,j}$ satisfying \eqref{si} with probability one, provided
\begin{equation}
\label{alignmentcondition}
2M-NK\geq d.
\end{equation}

With \eqref{pairwisealignment}, we see that in \eqref{zeroforcing_1}-\eqref{zeroforcing_2}, the three equations on the left are identical to the three equations on the right respectively. Ignoring the redundant equations, we rewrite \eqref{zeroforcing_1}-\eqref{zeroforcing_2} as
\begin{subequations}
\label{new_zeroforcing}
\begin{equation}
\label{new_zeroforcing1}
\sum_{k=0}^{K-1}\mathbf{V}_j\mathbf{G}_{j,k}\mathbf{F}_k\mathbf{H}_{k,j+1}\mathbf{U}_{j+1,j-1}=\mathbf{0},\quad\!\! j = 0,1,2.\\
\end{equation}
For completeness, we repeat \eqref{rank} as follows:
\begin{align}
\label{new_rank}
\mathrm{rank}\left(\sum_{k=0}^{K-1}\mathbf{V}_j\mathbf{G}_{j,k}\mathbf{F}_k\mathbf{W}_{k,j}\right)=2d,\quad \!\!j=0,1,2.
\end{align}
\end{subequations}
What remains is to determine the value of $\left(\frac{M}{N}, d\right)$ to ensure the existence of $\{\mathbf{F}_k\}$ and $\{\mathbf{V}_j\}$ satisfying \eqref{new_zeroforcing}.

\subsection{Achievable DoF for Signal Alignment I}
Based on the proposed signal alignment in \eqref{pairwisealignment}, for user $j$, the other interference term in \eqref{received signal} is rewritten as
\begin{align}
\label{otherinterference}
\sum_{k = 0}^{K-1}\mathbf{G}_{j,k}\mathbf{F}_k\mathbf{H}_{k,j+1}\mathbf{U}_{j+1,j+2}\left(\mathbf{S}_{j+1,j+2}+\mathbf{S}_{j+2,j+1}\right)
\end{align}
spanning a subspace of $d$-dimension. To neutralize the above interference at every user, we first design $\{\mathbf{F}_k\}$ to neutralize a portion of the interference, and then design $\left\{\mathbf{V}_j\right\}$ to null the remaining part. Specifically, we split the beamforming matrix $\mathbf{U}_{j,j+1}$ as
\begin{align}
\label{divideU}
\mathbf{U}_{j,j+1} = \left[\mathbf{U}^{(L)}_{j,j+1},\mathbf{U}^{(R)}_{j,j+1}\right], \quad\!\!j=0,1,2
\end{align}
where $\mathbf{U}^{(L)}_{j,j+1}\in \mathbb{C}^{M \times d'}$ and $\mathbf{U}^{(R)}_{j,j+1}\in \mathbb{C}^{M \times (d-d')}$ with $d'$ defined as
\begin{align}
\label{d'}
d' = 3d - M.
\end{align}
We design $\{\mathbf{F}_k\}$ to neutralize the interference corresponding to $\mathbf{U}_{j+1,j+2}^{(L)}$ for each user. That is, $\{\mathbf{F}_k\}$ need to satisfy
\begin{equation}
\label{pneutralize}
\sum_{k=0}^{K-1}\mathbf{G}_{j,k}\mathbf{F}_k\mathbf{H}_{k,j+1}\mathbf{U}^{(L)}_{j+1,j+2}=\mathbf{0},\quad\!\! j =0,1,2.\\
\end{equation}
With \eqref{pneutralize}, the received signal at user $j$ in \eqref{received signal} still contains a $2d$-dimensional desired signal and the $(d-d')$-dimensional interference corresponding to $\mathbf{U}^{(R)}_{j+1,j+2}$. With $d'$ in \eqref{d'}, we have $2d+d-d'=M$, implying that the antennas at each user are already fully utilized. Thus, the rank requirement \eqref{new_rank} can be alternatively represented by
\begin{small}
\begin{equation}
\label{prank}
\mathrm{rank}\left(\sum^{K-1}_{k=0}\mathbf{G}_{j,k}\mathbf{F}_k\left[\mathbf{W}_{k,j},\mathbf{H}_{k,j+1}\mathbf{U}^{(R)}_{j+1,j+2}\right]\right) =M,\quad\!\! j=0,1,2.\\
\end{equation}
\end{small}
\!\!\!We design $\mathbf{V}_j$ to null the remaining interference corresponding to $\mathbf{U}^{(R)}_{j+1,j+2}$ for each user. Specifically, we require $\mathbf{V}_{j} \in \mathbb{C}^{2d \times M}$ to satisfy
\begin{small}
\begin{equation}
\label{constructV}
\mathbf{V}_j\sum_{k=0}^{K-1}\mathbf{G}_{j,k}\mathbf{F}_k\mathbf{H}_{k,j+1}\mathbf{U}^{(R)}_{j+1,j+2}\!=\!\mathbf{0}, \quad \!\! j=0,1,2.
\end{equation}
\end{small}
\!\!Note that the left null space of $\sum_{k=0}^{K-1}\mathbf{G}_{j,k}\mathbf{F}_k\mathbf{H}_{k,j+1}\mathbf{U}^{(R)}_{j+1,j+2}\in \mathbb{C}^{M\times (d-d')}$ has at least $M-(d-d') =2d$ dimensions. Therefore, there always exists $\mathbf{V}_j \in \mathbb{C}^{2d \times M}$ with full row rank satisfying \eqref{constructV}. 

We now show that the above $\{\mathbf{F}_k\}$ and $\{\mathbf{V}_j\}$ satisfy \eqref{new_zeroforcing}. We start with \eqref{new_zeroforcing1} for $j=0$:
\begin{subequations}
\begin{small}
\begin{align}
\mathbf{V}_0\sum_{k=0}^{K-1}\mathbf{G}_{0,k}\mathbf{F}_k\mathbf{H}_{k,1}\mathbf{U}_{1,2}=& \left[\mathbf{V}_0\sum_{k=0}^{K-1}\mathbf{G}_{0,k}\mathbf{F}_k\mathbf{H}_{k,1}\mathbf{U}^{(L)}_{1,2},\quad \mathbf{V}_0\sum_{k=0}^{K-1}\mathbf{G}_{0,k}\mathbf{F}_k\mathbf{H}_{k,1}\mathbf{U}^{(R)}_{1,2}\right] \\
\label{constructV2_1}
=&\left[\mathbf{0},\quad\!\!\! \mathbf{V}_0\sum_{k=0}^{K-1}\mathbf{G}_{1,k}\mathbf{F}_k\mathbf{H}_{k,1}\mathbf{U}^{(R)}_{1,2}\right] \\
\label{constructV2_2}
=&\quad \!\!\!\!\mathbf{0}
\end{align}
\end{small}
\end{subequations}
\!\!where \eqref{constructV2_1} and \eqref{constructV2_2} follow from \eqref{pneutralize} and \eqref{constructV}, respectively. Similarly, condition \eqref{new_zeroforcing1} also holds for $j=1,2$. We then consider \eqref{new_rank} with $j=0$ and obtain
\begin{small}
\begin{subequations}
\begin{align}
\label{constructV3_1}
\mathrm{rank}\left(\mathbf{V}_0\sum_{k=0}^{K-1}\mathbf{G}_{0,k}\mathbf{F}_k\mathbf{W}_{k,0} \right)
=&\quad \!\!\!\mathrm{rank}\left(\left[\mathbf{V}_0\sum_{k=0}^{K-1}\mathbf{G}_{0,k}\mathbf{F}_k\mathbf{W}_{k,0},\mathbf{V}_0\sum_{k=0}^{K-1}\mathbf{G}_{0,k}\mathbf{F}_k\mathbf{H}_{k,1}\mathbf{U}^{(R)}_{1,2}\right]\right)\\
\label{constructV3_2}
=&\quad \!\!\!\mathrm{rank}\left(\mathbf{V}_0\sum_{k=0}^{K-1}\mathbf{G}_{0,k}\mathbf{F}_k\left[\mathbf{W}_{k,0}, \mathbf{H}_{k,1}\mathbf{U}^{(R)}_{1,2}\right]\right)=2d
\end{align}
\end{subequations}
\end{small}
\!\!where \eqref{constructV3_1} follows from \eqref{constructV}; \eqref{constructV3_2} is due to the fact that $\sum_{k=0}^{K-1}\mathbf{G}_{0,k}\mathbf{F}_k\left[\mathbf{W}_{k,0}, \mathbf{H}_{k,1}\mathbf{U}^{(R)}_{1,2}\right]\in \mathbb{C}^{M\times M}$ is of full rank from \eqref{prank}. Similarly, \eqref{new_rank} also holds for $j=1,2$.

Consequently, to prove that a certain point $\left(\frac{M}{N}, d\right)$ is achievable for Signal Alignment I, it suffices to show that there exist $\{\mathbf{U}_{j,j'}\}$ and $\{\mathbf{F}_k\}$ satisfying \eqref{pairwisealignment}, \eqref{pneutralize}, and \eqref{prank}. We have the following lemma:
\begin{lemma}
\label{LemmaDoF1}
For ${\frac{M}{N} \in \Big[\frac{9K+\sqrt{81K^2+60K}}{30},\frac{3K+\sqrt{9K^2-12K}}{6}\Big)}$ and $d<\frac{M}{3} + \frac{KN^2}{9M}$, there exist $\{\mathbf{U}_{j,j'}\}$ and $\{\mathbf{F}_k\}$ satisfying \eqref{pairwisealignment}, \eqref{pneutralize}, and \eqref{prank} with probability one.
\end{lemma}
\begin{proof}
See Appendix I-A.
\end{proof}

Note that the DoF $d$ in Lemma \ref{LemmaDoF1} is not necessarily an integer. A frequently used technique to achieve a rational DoF is symbol extension in which an extended MIMO system is constructed with multiple channel uses. The details can be found in Appendix II.

\subsection{Signal Alignment II}
\label{Signal Alignment II}
The derivation in the preceding subsection is based on the signal alignment \eqref{pairwisealignment}. However, this alignment imposes the requirement of \eqref{alignmentcondition} that may not always be met. In this subsection, we generalize the signal alignment in \eqref{pairwisealignment} by allowing the signal to be aligned in a subspace of the receiving signal space seen at each relay. Note that for Signal Alignment I in \eqref{pairwisealignment}, the full received signal space at each relay is considered in aligning signals. For  Signal Alignment II described below, we will only consider the projection of user signals into a subspace of the signal space at each relay. This will relax the signal alignment constraint in \eqref{alignmentcondition} at the cost of certain loss of freedom for each relay to process its received signal.

Let
\begin{equation}
\label{N'}
N' = \frac{2M-d}{K}
\end{equation}
be the dimension of the subspace for signal alignment. Of course, $d$ should be chosen to ensure $0\leq N' \leq N$. Then, each relay deactivates $N-N'$ antennas for signal reception, with the number of active antennas intact for broadcasting. That is, the number of active antennas of each relay is $N'$ in the uplink phase and $N$ in the downlink phase. This implies that the proposed uplink-downlink precoding design is asymmetric in nature.

A convenient way to select a subspace of dimension $N'$ is to deactivate the last $N-N'$ antennas at each relay. Then the precoding matrix of each relay $k$ can be decomposed as
\begin{equation}
\label{deactivate}
\mathbf{F}_k = \tilde{\mathbf{F}}_k \mathbf{E},\quad k = 0,\cdots,K-1
\end{equation}
where $\tilde{\mathbf{F}}_k \in \mathbb{C}^{N\times N'}$ and $\mathbf{E} = \left[\mathbf{I}_{N'}, \mathbf{0}_{N'\times (N-N')}\right] \in \mathbb{C}^{N'\times N}$. By the deactivation, the effective channel matrix from user $j$ to relay $k$ is $\tilde{\mathbf{H}}_{k,j} = \mathbf{E}\mathbf{H}_{k,j}\in \mathbb{C}^{N'\times M}$ and the processing matrix of relay $k$ becomes $\tilde{\mathbf{F}}_k$. As analogous to \eqref{pairwisealignment}, we aim to design $\mathbf{U}_{j,j'} \in \mathbb{C}^{M \times d}$ satisfying
\begin{align}
\label{alignment}
    \tilde{\mathbf{H}}_{k,j}\mathbf{U}_{j,j+1}  = \tilde{\mathbf{H}}_{k,j+1}\mathbf{U}_{j+1,j}, \quad \!\! k = 0,\cdots,K-1,
\end{align}
or equivalently
\begin{subequations}
\begin{align}
\tilde{\mathbf{K}}_j \underbrace{\left[\mathbf{U}_{j,j+1}^T, \mathbf{U}_{j+1,j}^T\right]^T}_{2M\times d}  = \mathbf{0}
\end{align}
where
\begin{align}
\label{K'_j}
\tilde{\mathbf{K}}_j = \left[\begin{array}{cc}
    \tilde{\mathbf{H}}_{0,j} & - \tilde{\mathbf{H}}_{0,j+1}\\
    \vdots & \vdots\\
    \tilde{\mathbf{H}}_{K-1,j} & - \tilde{\mathbf{H}}_{K-1,j+1}
  \end{array}\right] \in \mathbb{C}^{KN' \times 2M}, \quad \!\! j = 0,1,2.
\end{align}
\end{subequations}
Since $2M - KN' = d$ by the definition of $N'$ in \eqref{N'},
there exist full-rank matrices $\mathbf{U}_{j,j+1}$ and $\mathbf{U}_{j+1,j}$ satisfying \eqref{alignment} with probability one.

Similarly to Signal Alignment I, with \eqref{deactivate} and \eqref{alignment}, \eqref{zeroforcing_1}-\eqref{zeroforcing_2} reduces to
\begin{subequations}
\label{new_zeroforcing'}
\begin{align}
\label{new_zeroforcing1'}
\sum_{k=0}^{K-1}\mathbf{V}_j\mathbf{G}_{j,k}\tilde{\mathbf{F}}_k\tilde{\mathbf{H}}_{k,j+1}\mathbf{U}_{j+1,j+2}=\mathbf{0}
\end{align}
and \eqref{new_rank} can be rewritten as
\begin{align}
\label{newrank'}
\mathrm{rank}\left(\sum_{k=0}^{K-1}\mathbf{V}_j\mathbf{G}_{j,k}\tilde{\mathbf{F}}_k\tilde{\mathbf{W}}_{k,j}\right)=2d
\end{align}
\end{subequations}
where
\begin{equation}
\label{definew}
\tilde{\mathbf{W}}_{k,j} = \mathbf{E}\mathbf{W}_{k,j}, \quad j = 0,1,2.
\end{equation}
Then, what remains is to determine $\left(\frac{M}{N},d\right)$ that ensures the existence of $\{\tilde{\mathbf{F}}_k\}$ and $\{\mathbf{V}_j\}$ satisfying \eqref{new_zeroforcing'}.

\subsection{Achievable DoF for Signal Alignment II}
We now consider the construction of $\{\tilde{\mathbf{F}}_k,\mathbf{V}_j\}$ to satisfy \eqref{new_zeroforcing'}. Similarly to Signal Alignment I, $\{\tilde{\mathbf{F}}_k\}$ are designed to neutralize the interference corresponding to $\mathbf{U}^{(L)}_{j,j+1} \in \mathbb{C}^{M \times d'}$, while $\{\mathbf{V}_j\}$ are designed to null the remaining interference corresponding to $\mathbf{U}^{(R)}_{j,j+1} \in \mathbb{C}^{M \times (d-d')}$. More specifically, as analogous to \eqref{pneutralize} and \eqref{prank}, we require $\{\tilde{\mathbf{F}}_k\}$ to satisfy
\begin{subequations}
\label{pneutralize'}
\begin{align}
\label{pneutralize1'}
\sum_{k=0}^{K-1}\mathbf{G}_{j,k}\tilde{\mathbf{F}}_k\tilde{\mathbf{H}}_{k,j+1}\mathbf{U}^{(L)}_{j+1,j+2}=\mathbf{0}
\end{align}
\begin{align}
\label{pneutralize4'}
\mathrm{rank}\left(\sum^{K-1}_{k=0}\mathbf{G}_{j,k}\tilde{\mathbf{F}}_k\left[\tilde{\mathbf{W}}_{k,j},\tilde{\mathbf{H}}_{k,j+1}\mathbf{U}^{(R)}_{j+1,j+2}\right]\right) &=M , \quad\!\! j=0,1,2.
\end{align}
\end{subequations}
Moreover, $\mathbf{V}_{j} \in \mathbb{C}^{2d \times M}$ is of full row rank and satisfies
\begin{equation}
\label{constructV'}
\mathbf{V}_j\sum_{k=0}^{K-1}\mathbf{G}_{j,k}\tilde{\mathbf{F}}_k\tilde{\mathbf{H}}_{k,j+1}\mathbf{U}^{(R)}_{j+1,j+2}\!=\!\mathbf{0}, \quad\!\! j=0,1,2.
\end{equation}
It can be readily shown that, with $d'$ in \eqref{d'} and $N'$ in \eqref{N'}, there always exist $\{\mathbf{V}_j\}$ satisfying \eqref{constructV'}. Thus, to prove the achievability of a certain DoF point $\left(\frac{M}{N},d\right)$, it suffices to show that there exist $\{\tilde{\mathbf{F}}_k\}$ satisfying \eqref{pneutralize'}. We have the following result.

\begin{lemma}
\label{LemmaDoF2}
For ${\frac{M}{N} \in \Big[1, \frac{9K+\sqrt{81K^2+60K}}{30}\Big)}$ and $d<\frac{3M^2+2MN}{9M+N}$, there exist $\{\mathbf{U}_{j,j'}\}$ and $\{\tilde{\mathbf{F}}_k\}$ satisfying \eqref{alignment} and \eqref{pneutralize'} with probability one.
\end{lemma}
\begin{proof}
See Appendix I-B.
\end{proof}

\subsection{Achievable DoF with No Signal Alignment}
From Lemmas \ref{LemmaDoF1} and \ref{LemmaDoF2}, we see that both Signal Alignment I and II cannot be realized when $\frac{M}{N}$ is relatively small. We next focus on relatively small values of $\frac{M}{N}$ and show that $d_\mathrm{sum}=\frac{3M}{2}$ (or equivalently, $d=\frac{M}{2}$) is achievable for $\frac{M}{N} \in \left(0,\frac{\sqrt{3K}}{K}\right)$. To start with, we set
\begin{equation}
\label{UV}
\mathbf{V}_{j} = \mathbf{I}_{M}, \quad\!\!
\mathbf{U}_{j,j+1}  =  [\mathbf{I}_{d} ,\mathbf{0}_{d\times d}]^T,\quad\!\! \mathbf{U}_{j,j-1}  = [\mathbf{0}_{d\times d} ,\mathbf{I}_{d}]^T, \quad\!\!\mathrm{for} \quad\!\! j =0,1,2.
\end{equation}
We see that, \eqref{zeroforcing_1}-\eqref{zeroforcing_2} reduce to a linear system of $\{\mathbf{F}_k\}$ with $KN^2$ unknown variables and $3M^2$ equations. The system allows non-zero solutions of $\{\mathbf{F}_k\}$ provided $KN^2 > 3M^2$, i.e., $\frac{M}{N} \in \left(0, \frac{\sqrt{3K}}{3}\right)$. To prove the achievability of $d = \frac{M}{2}$, it suffices to show that there exist $\{\mathbf{F}_k\}$ satisfying all the conditions in \eqref{zeroforcing}. We have the following result. 

\begin{lemma}
\label{LemmaDoF3}
For ${\frac{M}{N} \in \Big(0,\frac{\sqrt{3K}}{K}\Big)}$ and $d = \frac{M}{2}$, there exist $\{\mathbf{F}_{k}\}$, together with $\{\mathbf{U}_{j,j'}, \mathbf{V}_{j}\}$ in \eqref{UV}, satisfying \eqref{zeroforcing} with probability one.
\end{lemma}
\begin{proof}
See Appendix I-C.
\end{proof}

\subsection{Achievable DoF Using Antenna Disablement}
In the preceding subsections, we have established an achievable DoF for $\frac{M}{N} \in \Big[0, \frac{\sqrt{3K}}{3}\Big) \cup \Big[1,\frac{3K+\sqrt{9K^2-12K}}{6}\Big)$. We now follow the antenna disablement approach \cite{DOF2} to establish the achievable DoF for other ranges of $\frac{M}{N}$. Specifically, for $\frac{M}{N} \in \left(0,1\right)$, we disable $N-M$ antennas at each relay. Then, from Lemma \ref{LemmaDoF2}, we see that any DoF $d<\frac{3M^2+2MN^*}{9M+N^*} = \frac{M}{2}$ can be achieved, where $N^*=M$ is the number of active antennas at each relay. The only issue is that $d$ may be not an integer. This issue can be solved by the technique of symbol extension described in Appendix II. Similarly, with symbol extension and antenna disablement, $d<\frac{KN}{3}$ is achievable for $\frac{M}{N} \in \left(\frac{3K+\sqrt{9K^2-12K}}{6},\infty\right)$. 

Combining Lemmas \ref{LemmaDoF1}, \ref{LemmaDoF2}, and \ref{LemmaDoF3}, we conclude that any DoF $d$ satisfying $d < d^*$ is achievable, where
\begin{equation}
d^* =
\begin{cases}
\frac{M}{2}  &        {\frac{M}{N} \in \Big[0,\max\left\{\frac{\sqrt{3K}}{3},1\right\}\Big)}\\
\max\left\{\frac{M}{3} + \frac{5MN}{27M+3N},\frac{\sqrt{3K}N}{6}\right\}      & {\frac{M}{N} \in \Big[\max\left\{\frac{\sqrt{3K}}{3},1\right\}, \frac{9K+\sqrt{81K^2+60K}}{30}\Big)}  \\
\frac{M}{3} + \frac{KN^2}{9M}     & {\frac{M}{N} \in \Big[\frac{9K+\sqrt{81K^2+60K}}{30},\frac{3K+\sqrt{9K^2-12K}}{6}\Big)} \\
\frac{KN}{3}  &      {\frac{M}{N} \in \Big[\frac{3K+\sqrt{9K^2-12K}}{6},\infty\Big)}.
\end{cases}
\end{equation}
With $d_\mathrm{sum}=3d$, we immediately obtain \eqref{achievableDoF0}. This completes  the proof of Theorem 1.

\section{Conclusion and Future Work}
In this paper, we developed a new formalism to analyze the achievable DoF of the symmetric multi-relay MIMO Y channel. Specifically, we adopted the idea of uplink-downlink asymmetric design and proposed a new method to tackle the solvability problem of linear systems with rank constraints. In the proposed design, we also incorporated the techniques of signal alignment, antenna disablement, and symbol extension. An achievable DoF for an arbitrary configuration of $(M, N, K)$ was derived. 

The study of multi-relay MIMO mRCs is still in an initial stage. Based on our work, the following directions will be of interest for future research.

\subsection{Tighter Upper Bounds}
For $\frac{M}{N} \in \Big(\max\left\{\frac{\sqrt{3K}}{3},1\right\}, \frac{3K+\sqrt{9K^2-12K}}{6}\Big)$, our achievable total DoF does not match the full relay-cooperation upper bound in \eqref{upperbound}. We conjecture that the main reason for this mismatch is that the upper bound is too loose in this range of $\frac{M}{N}$. As such, tighter upper bounds are highly desirable to fully characterize the DoF of the symmetric multi-relay MIMO Y channel. This, however, requires careful analysis on the fundamental performance degradation caused by the separation of relays.

\subsection{General Antenna and DoF Setups}
In this paper, we considered the symmetric multi-relay MIMO Y channel, where the numbers of antennas of all user nodes are assumed to be the same. We also assumed a symmetric DoF setting, where each user transmits the same number of independent spatial data streams. The main purpose of these assumptions is to avoid the combinatorial complexity in manipulating signals and interference. The techniques used in this paper can be extended to the cases with asymmetric antenna and DoF setups. However, this results in a DoF achievability problem far more complicated than \eqref{zeroforcing}, since we need to analyze the feasibility of all possible DoF tuples of $(d_{0,1},d_{1,0},d_{1,2},d_{2,1},d_{0,2},d_{2,0})$ under an asymmetric antenna configuration. The optimal DoF region for the single-relay case has been recently reported in \cite{generalMIMOY}. We believe that the techniques used in \cite{generalMIMOY} will provide some insights on deriving the DoF region of the multi-relay case. 

\subsection{Cases with More Users}
Our approach can be extended to multi-relay MIMO Y channels with more than three users. However, we emphasize that such an extension is not trivial. As seen from \cite{Gao, multirelay, MIMO2, yuan}, for MIMO mRCs with more than three users, more signal alignment patterns than pairwise alignment should be exploited to support efficient data exchanges. This implies that in multi-relay MIMO Y channels with more than three users, we need to combine our uplink-downlink asymmetric approach with more intelligent signal alignment strategies. Therefore, the extension to multi-relay MIMO Y channels with more users will be an interesting research topic worthy of future effort.

\begin{appendices}
\section{Proof of Lemmas \ref{LemmaDoF1}-\ref{LemmaDoF3}}
\label{Proof of Lemma 1 and Lemma 2}
\subsection{Proof of Lemma \ref{LemmaDoF1}}
\label{Proof of Lemma 1}
We need to prove that for $\frac{M}{N} \in \Big[\frac{9K+\sqrt{81K^2+60K}}{30}, \frac{3K+\sqrt{9K^2-12K}}{6}\Big)$ and $d<\frac{M}{3} + \frac{KN^2}{9M}$, there exist $\{\mathbf{U}_{j,j'}\}$ and $\{\mathbf{F}_k\}$ satisfying \eqref{pairwisealignment}, \eqref{pneutralize}, and \eqref{prank} with probability one. The main steps of the proof are presented as follows:
\begin{itemize}
\item Show that the signal alignment \eqref{pairwisealignment} can be performed and $d'$ in \eqref{d'} is well-defined.
\item For $M,N$ and $d$ in the given ranges, construct a set $\mathcal{T}_{M,N,d}$ of channel realizations such that a randomly generated channel realization belongs to $\mathcal{T}_{M,N,d}$ with probability one.
\item Prove that for almost all elements in $\mathcal{T}_{M,N,d}$, there exist $\{\mathbf{U}_{j,j'}\}$ and $\{\mathbf{F}_k\}$ satisfying \eqref{pairwisealignment}, \eqref{pneutralize}, and \eqref{prank}. \footnote{Although the term ``satisfying (24)" appears both here and in (54), the conditions are different. In (54), we require that there exist $\{\mathbf{U}_{j,j'}\}$ satisfying (24) and $\mathrm{rank}(\mathbf{K})=3d'M$, while here we require that there exist $\{\mathbf{U}_{j,j'}\}$ and $\{\mathbf{F}_k\}$ satisfying (24), (31), and (32). It is possible that for some channel realization, there exist $\{\mathbf{U}_{j,j'}\}$  satisfying (24), but there do not exist $\{\mathbf{F}_k\}$, together with these $\{\mathbf{U}_{j,j'}\}$, satisfying (31) and (32).}
\end{itemize}

We first show that the signal alignment \eqref{pairwisealignment} can be performed. For $\frac{M}{N} \geq \frac{9K+\sqrt{81K^2+60K}}{30}$, we have
{\setlength{\abovedisplayskip}{4pt}
\setlength{\belowdisplayskip}{4pt}
\begin{equation}
2M- \left(\frac{M}{3} + \frac{KN^2}{9M}\right) \geq NK.
\end{equation}}
\!\!Further, as $d<\frac{M}{3} + \frac{KN^2}{9M}$, we obtain 
{\setlength{\abovedisplayskip}{4pt}
\setlength{\belowdisplayskip}{4pt}
\begin{equation}
2M -d > 2M- \left(\frac{M}{3} + \frac{KN^2}{9M}\right) \geq NK.
\end{equation}}
\!\!Therefore, \eqref{alignmentcondition} is met, and so there exist full-column-rank $\{\mathbf{U}_{j,j'}\}$ satisfying \eqref{pairwisealignment} with probability one.

We next show that $d'$ in \eqref{d'} is well-defined. That is, $0 \leq d' = 3d-M \leq d$ holds for $d$ chosen sufficiently close to $\frac{M}{3}+\frac{KN^2}{9M}$. For $\frac{M}{N} \in \left[\frac{9K+\sqrt{81K^2+60K}}{30}, \frac{3K+\sqrt{9K^2-12K}}{6} \right]$, together with $d<\frac{M}{3} + \frac{KN^2}{9M}$ and $K\geq2$, we obtain
\begin{equation}
d'-d = 2d -M <\frac{2M}{3} + \frac{2KN^2}{9M} - M = \frac{2KN^2 - 3M^2}{9M} < 0,
\end{equation}
where the last step holds by noting $\frac{M}{N}>\frac{3}{5}K > \frac{\sqrt{2K}}{3}$ for $K \geq 2$. On the other hand, as
\begin{equation}
3\times \left(\frac{M}{3} + \frac{KN^2}{9M}\right) -M =\frac{KN^2}{3M} >0,
\end{equation}
we can always choose $d$ close to $\frac{M}{3} + \frac{KN^2}{9M}$ to ensure $d' = 3d-M>0$. We henceforth always assume that $d$ is appropriately chosen such that $0\leq d' \leq d$.

We now consider step 2 of the proof. Denote the overall channel $\mathbf{T}$ of the symmetric multi-relay MIMO Y channel by
{\setlength{\abovedisplayskip}{4pt}
\setlength{\belowdisplayskip}{4pt}
\begin{equation}
\label{channel_realization}
\mathbf{T} = \left(\mathbf{H},\mathbf{G}\right) \in \mathbb{C}^{KN\times 3M} \times \mathbb{C}^{3M \times KN}
\end{equation}}
\!\!where
{\setlength{\abovedisplayskip}{4pt}
\setlength{\belowdisplayskip}{4pt}
\begin{equation}
\label{uplink_channel_realization}
\mathbf{H} = \left[\begin{array}{ccc}
    \mathbf{H}_{0,0} &  \mathbf{H}_{0,1} & \mathbf{H}_{0,2}\\
    \vdots & \vdots & \vdots\\
    \mathbf{H}_{K-1,0} & \mathbf{H}_{K-1,1} & \mathbf{H}_{K-1,2}\\
\end{array}
\right] ,\quad
\mathbf{G} = \left[\begin{array}{cccc}
    \mathbf{G}_{0,0} &\cdots& \mathbf{G}_{0,K-1}\\
    \mathbf{G}_{1,0} &\cdots& \mathbf{G}_{1,K-1}\\
    \mathbf{G}_{2,0} &\cdots& \mathbf{G}_{2,K-1}\\
\end{array}
\right].
\end{equation}}
\!\!Then, we rewrite \eqref{pneutralize} using Kronecker product as
{\setlength{\abovedisplayskip}{4pt}
\setlength{\belowdisplayskip}{4pt}
\begin{equation}
\label{Kform}
\mathbf{K}\mathbf{f} = \mathbf{0},
\end{equation}}
\!\!where
\begin{subequations}
\begin{equation}
\label{K}
\mathbf{K}\!=\!\!\left[\!\!\!
\begin{array}{cccc}
    \left(\mathbf{H}_{0,1}\mathbf{U}^{(L)}_{1,2}\right)^{T}\!\!\! \otimes\! \mathbf{G}_{0,0} \!&\!  \left(\mathbf{H}_{1,1}\mathbf{U}^{(L)}_{1,2}\right)^{T}\!\!\!\otimes \!\mathbf{G}_{0,1} \!&\! \!\cdots\! \!&\!  \left(\mathbf{H}_{K-1,1}\mathbf{U}^{(L)}_{1,2}\right)^{T}\!\!\!\otimes\! \mathbf{G}_{0,K-1}\\
       \left(\mathbf{H}_{0,2}\mathbf{U}^{(L)}_{2,0}\right)^{T}\!\!\!\otimes\! \mathbf{G}_{1,0} \!&\! \left(\mathbf{H}_{1,2}\mathbf{U}^{(L)}_{2,0}\right)^{T}\!\!\otimes\! \mathbf{G}_{1,1} \!&\! \!\cdots\! \!&\! \left(\mathbf{H}_{K-1,2}\mathbf{U}^{(L)}_{2,0}\right)^{T}\!\!\!\otimes \!\mathbf{G}_{1,K-1}\\
      \left(\mathbf{H}_{0,0}\mathbf{U}^{(L)}_{0,1}\right)^{T}\!\!\!\otimes \!\mathbf{G}_{2,0} \!&\! \left(\mathbf{H}_{1,0}\mathbf{U}^{(L)}_{0,1}\right)^{T}\!\!\!\otimes \!\mathbf{G}_{2,1} \!&\! \!\cdots\!\! &\! \left(\mathbf{H}_{K-1,0}\mathbf{U}^{(L)}_{0,1}\right)^{T}\!\!\!\otimes\! \mathbf{G}_{2,K-1}
\end{array}
\!\!\right]\!\!\in \mathbb{C}^{3d'M\times KN^2}
\end{equation}
\begin{equation}
\mathbf{f} = \left[
\begin{array}{cccc}
\mathrm{vec}(\mathbf{F}_0)^T&
\mathrm{vec}(\mathbf{F}_1)^T&
\cdots&
\mathrm{vec}(\mathbf{F}_{K-1})^T
\end{array}
\right]^T \in \mathbb{C}^{KN^2 \times 1}.
\end{equation}
\end{subequations}

We are now ready to define the set $\mathcal{T}_{M,N,d}$. For $\frac{M}{N} \in \Big[\frac{9K+\sqrt{81K^2+60K}}{30}, \frac{3K+\sqrt{9K^2-12K}}{6}\Big)$ and $d<\frac{M}{3} + \frac{KN^2}{9M}$, define 
{\setlength{\abovedisplayskip}{3pt}
\setlength{\belowdisplayskip}{4pt}
\begin{equation}
\label{DefTMNd}
\mathcal{T}_{M,N,d} = \left\{\mathbf{T}\quad\!\!\left|\quad\!\!
\begin{aligned}
&\text{All $\mathbf{T}$ satisfying:} \\[-0.3cm]
&\text{1) }\mathrm{rank}\left(\mathbf{K}_{j}\right) = KN,\forall j;\\[-0.3cm]
&\text{2) }\text{there exist $\{\mathbf{U}_{j,j'}\}$ satisfying \eqref{pairwisealignment} and $\mathrm{rank}\left(\mathbf{K}\right) = 3d'M$}
\end{aligned}\right.
\right\}
\end{equation}}
\!\!\!where $\mathbf{K}_j$ is defined in \eqref{K_j}. We claim that a randomly generated $\mathbf{T}$ belongs to $\mathcal{T}_{M,N,d}$ with probability one. Recall that the entries of $\mathbf{T}$ are drawn independently from a continuous distribution. Since $\mathbf{K}_{j}$ is a wide matrix, it is of full row rank ($=KN$) with probability one. We next show that for a random $\mathbf{T}$ and full-column-rank $\{\mathbf{U}_{j,j'}\}$ satisfying \eqref{pairwisealignment}, $\mathbf{K}$ in \eqref{K} is of full row rank with probability one. To see this, we first note that $\mathbf{K}$ is a wide matrix since
{\setlength{\abovedisplayskip}{3pt}
\setlength{\belowdisplayskip}{3pt}
\begin{align}
\label{wideK}
KN^2-3d'M & = KN^2+3M^2-9dM > KN^2+3M^2 - 9M\left(\frac{M}{3} + \frac{KN^2}{9M}\right) = 0.
\end{align}}
\!\!Second, from the channel randomness, we have $\mathrm{rank}\left(\mathbf{H}_{k,j}\mathbf{U}^{(L)}_{j,j'}\right)=d'$ and $\mathrm{rank}\left(\mathbf{G}_{j',k}\right)=N$. Then $\left(\mathbf{H}_{k,j}\mathbf{U}^{(L)}_{j,j'}\right)^{T}\!\!\otimes \mathbf{G}_{j',k}$ is of rank $d'N$. Each $d'M \times KN^2$ block-row of $\mathbf{K}$ consists of $K$ submatrices in the form of $\left(\mathbf{H}_{k,j}\mathbf{U}^{(L)}_{j,j'}\right)^{T}\!\!\otimes \mathbf{G}_{j',k}$. From the channel randomness, the rank of each block-row is given by $\min\left\{d'M, Kd'N\right\} = d'M$, since $\frac{M}{N} \leq \frac{3K+\sqrt{9K^2-12K}}{6}<K$. Further, by noting that the three block-rows of $\mathbf{K}$ are statistically independent of each other, we conclude that $\mathbf{K}$ is of full row rank ($=3d'M$) with probability one. 

We now consider the last step, i.e., to show that for almost all $\mathbf{T}$ in $\mathcal{T}_{M,N,d}$, there exist $\{\mathbf{U}_{j,j'}\}$ and $\{\mathbf{F}_k\}$ satisfying \eqref{pairwisealignment}, \eqref{pneutralize}, and \eqref{prank}. To proceed, we present a useful lemma below.

\begin{lemma}
\label{ProveLemma1}
For $d <\frac{M}{3} + \frac{KN^2}{9M}$ and $\frac{M}{N} \in \Big[\frac{9K+\sqrt{81K^2+60K}}{30}, \frac{3K+\sqrt{9K^2-12K}}{6}\Big)$, assume that there exist a certain element $\widehat{\mathbf{T}} \in \mathcal{T}_{M,N,d}$, full-row-rank $\{\widehat{\mathbf{U}}_{j,j'}\}$, and relay processing matrices $\{\widehat{\mathbf{F}}_k\}$ such that \eqref{pairwisealignment}, \eqref{pneutralize}, and \eqref{prank} hold.
Then for random $\mathbf{T} \in \mathcal{T}_{M,N,d}$, there exist $\{\mathbf{U}_{j,j'}\}$ and $\{\mathbf{F}_k\}$ satisfying \eqref{pairwisealignment}, \eqref{pneutralize}, and \eqref{prank} with probability one.
\end{lemma}
\begin{proof}
Consider a random $\mathbf{T} \in \mathcal{T}_{M,N,d}$. By the definition of $\mathcal{T}_{M,N,d}$, there exist $\{\mathbf{U}_{j,j'}\}$ satisfying \eqref{pairwisealignment} and $\mathrm{rank}\left(\mathbf{K}\right) = 3d'M$.  Using Gaussian elimination, the general solution of \eqref{Kform} can be written by
{\setlength{\abovedisplayskip}{4pt}
\setlength{\belowdisplayskip}{4pt}
\begin{align}
\label{general_solution}
\mathbf{f} = \sum_{a=0}^{A-1}\alpha_a \mathbf{f}_a,
\end{align}}
\!\!where $\alpha_1,\alpha_2,\cdots,\alpha_A$ are free variables with $A = KN^2-3d'M$, and $\mathbf{f}_a\in\mathbb{C}^{2NN'}$, $a = 0, \cdots, A-1$, span the right null space of $\mathbf{K}$.
Recall that Gaussian elimination involves the following arithmetic operations: addition, subtraction, multiplication and division. This implies that each entry of $\mathbf{f}_a$ is a rational function of the entries of $\{\mathbf{G}_{j,k}\}$ and $\{\mathbf{H}_{k,j}\mathbf{U}^{(L)}_{j,j'}\}$. Thus, with proper scaling, we can always express each entry of $\mathbf{f}$ as a finite-degree polynomial of the entries of  $\{\mathbf{G}_{j,k}\}$ and $\{\mathbf{H}_{k,j}\mathbf{U}^{(L)}_{j,j'}\}$. Following similar arguments, since $\{\mathbf{U}_{j,j'}\}$ are designed to satisfy $\eqref{K_j}$ and each $\mathbf{K}_j$ is of full row rank, entries of matrices $\{\mathbf{U}_{j,j'}\}$ can be represented by polynomials of the entries of $\{\mathbf{H}_{k,j}\}$ and $\{\beta_b\}$, where $\beta_0,\beta_1,\cdots,\beta_{B-1}$ are free variables with $B = 3(2M-KN)$. This implies that each entry of $\mathbf{f}$ satisfying \eqref{Kform} can be represented by a finite-degree polynomial of the entries of $\mathbf{T}$, $\{\alpha_a\}$, and $\{\beta_b\}$. Recall that \eqref{Kform} is the Kronecker product form of \eqref{pneutralize} and $\mathbf{f}$ consists of the vectorizations of $\{\mathbf{F}_k\}$. Therefore, entries of $\{\mathbf{F}_k\}$ satisfying \eqref{pneutralize} can be represented by finite-degree polynomials of the entries of $\mathbf{T}$, $\{\alpha_a\}$, and $\{\beta_b\}$.

Now consider the following determinant product:
{\setlength{\abovedisplayskip}{4pt}
\setlength{\belowdisplayskip}{4pt}
\begin{small}
\begin{align}
\label{D}
D_{M,N,d}=\prod_{j=0}^2\mathrm{det}\left(\sum_{k=0}^{K-1}\mathbf{G}_{j,k}\mathbf{F}_k\left[\mathbf{W}_{k,j},\mathbf{H}_{k,j+1}\mathbf{U}_{j+1,j+2}^{(R)}\right]\right)\!\!.
\end{align}
\end{small}}
\!\!\!Note that $D_{M,N,d}\not = 0$ if and only if the matrix $\sum_{k=0}^{K-1}\mathbf{G}_{j,k}\mathbf{F}_k\left[\mathbf{W}_{k,j},\mathbf{H}_{k,j+1}\mathbf{U}_{j+1,j+2}^{(R)}\right]$ is of full rank for $j = 0,1,2$, or equivalently, \eqref{prank} holds. Therefore, to prove Lemma \ref{ProveLemma1}, it suffices to show that for a random $\mathbf{T}\in \mathcal{T}_{M,N,d}$, there exist $\{\mathbf{U}_{j,j'}\}$ and $\{\mathbf{F}_k\}$ satisfying \eqref{pairwisealignment} and $\eqref{pneutralize}$, such that $D_{M,N,d}\not =0$ with probability one. 

To this end, we first note that $D_{M,N,d}$ is a finite-degree polynomial with respect to entries of $\mathbf{T}$, $\{\alpha_a\}$, and $\{\beta_b\}$, i.e., $D_{M,N,d}$ can be represented by
{\setlength{\abovedisplayskip}{4pt}
\setlength{\belowdisplayskip}{4pt}
\begin{small}
\begin{align}
\label{D_polynomial}
D_{M,N,d} = \sum^{T_\mathrm{max}}_{t=1}p_t\left(\mathbf{T}\right)g_t\left(\{\alpha_a\},\{\beta_b\}\right)
\end{align}
\end{small}}
\!\!where $T_\mathrm{max}$ is a finite integer, $p_t\left(\cdot\right)$ is a polynomial of the entries of $\mathbf{T}$, and $g_t\left(\cdot, \cdot\right)$ is a monomial of $\{\alpha_a\}$, $\{\beta_b\}$. Note that $g_t(\cdot, \cdot) \not= g_{t'}(\cdot,\cdot)$ for $t\not=t'$. By assumption, there exist $\widehat{\mathbf{T}}$ and $\{\widehat{\mathbf{U}}_{j,k},\widehat{\mathbf{F}}_k\}$ satisfying \eqref{pairwisealignment}, \eqref{pneutralize}, and \eqref{prank}, implying that there exist $\widehat{\mathbf{T}}$, $\{\widehat{\alpha}_a\}$, and $\{\widehat{\beta}_b\}$ such that $D_{M,N,d} \not = 0$. Thus, $D_{M,N,d}$ is a non-zero polynomial. Let $\mathcal{I}$ be the index set such that $p_t(\cdot)$, $t \in \mathcal{I}$, is a non-zero polynomial. Denote by $\mathcal{V}$ the solution set of the polynomial system: $p_t(\mathbf{T})=0$, $t\in \mathcal{I}$. From algebraic geometry, $\mathcal{V}$ has Lebesgue measure zero in $\mathbb{C}^{KN\times 3M} \times \mathbb{C}^{3M \times KN}$. That is, for a random generated  $\mathbf{T}$, the probability of $\mathbf{T}\in \mathcal{V}$ is zero. Thus, for a random $\mathbf{T}\in \mathcal{T}_{M,N,d}$, there is at least one $p_t(\mathbf{T}) \not = 0$, and $D_{M,N,d}$ in \eqref{D_polynomial} is a non-zero polynomial of $\{\alpha_a\}$ and $\{\beta_b\}$ with probability one. Therefore, we can always find $\{\alpha_a\}$ and $\{\beta_b\}$ such that $D_{M,N,d} \not= 0$. That is, there exist $\{\mathbf{U}_{j',j}\}$ and $\{\mathbf{F}_k\}$ satisfying \eqref{pairwisealignment}, \eqref{pneutralize}, and \eqref{prank}, which concludes the proof of Lemma \ref{ProveLemma1}.
\end{proof}

To invoke Lemma \ref{ProveLemma1}, we need to show the existence of a certain $\widehat{\mathbf{T}}\in \mathcal{T}_{M,N,d}$, $\{\widehat{\mathbf{U}}_{j,j'}\}$, and $\{\widehat{\mathbf{F}}_k\}$ satisfying \eqref{pairwisealignment}, \eqref{pneutralize}, and \eqref{prank}. To this end, we set $\widehat{\mathbf{F}}_k = \mathbf{I}_N$ for $k=0,\cdots, K-1$, and choose $\{\widehat{\mathbf{G}}_{j,k}\}$ and $\{\widehat{\mathbf{H}}_{k,j}\widehat{\mathbf{U}}^{(R)}_{j,j+1}\}$ to be random matrices with the entries independently drawn from a continuous distribution. Then, we choose full-rank matrices $\{\widehat{\mathbf{H}}_{k,j}\widehat{\mathbf{U}}^{(L)}_{j,j+1}\}$ to satisfy
{\setlength{\abovedisplayskip}{3pt} 
\setlength{\belowdisplayskip}{5pt}
\begin{small}
\begin{align}
\label{specificH}
\underbrace{\left[\begin{array}{cccc}
    \widehat{\mathbf{G}}_{j,0} &  \widehat{\mathbf{G}}_{j,1} & \cdots & \widehat{\mathbf{G}}_{j,K-1}
\end{array}\right]}_{M \times KN} \underbrace{\left[\begin{array}{c} \widehat{\mathbf{H}}_{0,j+1} \widehat{\mathbf{U}}^{(L)}_{j+1,j+2}\\  \vdots \\  \widehat{\mathbf{H}}_{K-1,j+1} \widehat{\mathbf{U}}^{(L)}_{j+1,j+2} \end{array}\!\!\right]}_{KN\times d'}  = \mathbf{0}.
\end{align}
\end{small}}
\!\!It can be verified that $KN \geq 3d$ for $d<\frac{M}{3} + \frac{KN^2}{9M}$, $\frac{M}{N} \in \left[\frac{9K+\sqrt{81K^2+60K}}{30},\frac{3K+\sqrt{9K^2-12K}}{6}\right]$. Then
{\setlength{\abovedisplayskip}{3pt}
\setlength{\belowdisplayskip}{3pt}
\begin{equation}
\label{constructcondition}
KN-M \geq 3d-M = d',
\end{equation}}
\!\!implying that the null space of the $M \times KN$ matrix in \eqref{specificH} has at least $d'$ dimensions. Thus full-column-rank $\{\widehat{\mathbf{H}}_{k,j}\widehat{\mathbf{U}}^{(L)}_{j,j+1}\}$ satisfying \eqref{specificH} exist with probability one. Based on the chosen $\{\widehat{\mathbf{H}}_{k,j}\widehat{\mathbf{U}}^{(L)}_{j,j+1}\}$ and $\{\widehat{\mathbf{H}}_{k,j}\widehat{\mathbf{U}}^{(R)}_{j,j+1}\}$, we can readily determine the values of $\{\widehat{\mathbf{H}}_{k,j}\}$ and $\{\widehat{\mathbf{U}}_{j,j+1}\}$ (not necessarily unique). With $\eqref{pairwisealignment}$, $\{\widehat{\mathbf{U}}_{j+1,j}\}$ are also determined. Finally, $\widehat{\mathbf{T}}$ is determined by $\{\widehat{\mathbf{H}}_{k,j},\widehat{\mathbf{G}}_{j,k}\}$. It can be verified that the constructed $\widehat{\mathbf{T}}$ belongs to $\mathcal{T}_{M,N,d}$ with probability one. 

We now show that the above constructed $\widehat{\mathbf{T}}$, $\{\widehat{\mathbf{F}}_k\}$, and $\{\widehat{\mathbf{U}}_{j,j'}\}$ satisfy \eqref{pairwisealignment}, \eqref{pneutralize}, and \eqref{prank} with probability one. First, by construction, \eqref{pairwisealignment} is automatically met. Further, as $\widehat{\mathbf{F}}_k = \mathbf{I}_N$, condition \eqref{pneutralize} reduces to \eqref{specificH} which holds again by construction. Thus \eqref{pneutralize} holds. To check \eqref{prank}, it suffices to consider the case $j=0$ due to symmetry. Note that $\{\widehat{\mathbf{W}}_{k,0}\}$ are determined by $\{\widehat{\mathbf{H}}_{k,0}\widehat{\mathbf{U}}^{(L)}_{0,2}\}$ and $\{\widehat{\mathbf{H}}_{k,0}\widehat{\mathbf{U}}^{(L)}_{0,1}\}$, which only depend on $\{\widehat{\mathbf{G}}_{1,k}\}$ and $\{\widehat{\mathbf{G}}_{2,k}\}$. That is, $\{\widehat{\mathbf{W}}_{k,0}\}$ are not functions of $\{\widehat{\mathbf{G}}_{0,k}\}$. Moreover, both $\{\widehat{\mathbf{G}}_{0,k}\}$ and $\{\widehat{\mathbf{W}}_{k,0}\}$ are independent of $\{\widehat{\mathbf{H}}_{k,1}\widehat{\mathbf{U}}^{(R)}_{1,2}\}$. Therefore, 
{\setlength{\abovedisplayskip}{3pt}
\setlength{\belowdisplayskip}{3pt}
\begin{small}
\begin{equation}
\mathrm{rank}\left(\widehat{\mathbf{G}}_{0,k}\widehat{\mathbf{F}}_k\left[\widehat{\mathbf{W}}_{k,0},\widehat{\mathbf{H}}_{k,1}\widehat{\mathbf{U}}^{(R)}_{1,2}\right]\right) = \min\{M,N\} = N
\end{equation}
\end{small}}
\!\!with probability one, where the first equality follows from $\widehat{\mathbf{F}}_{k} \in \mathbb{C}^{N \times N}$, $\widehat{\mathbf{G}}_{0,k} \in \mathbb{C}^{M \times N}$, and $\left[\widehat{\mathbf{W}}_{k,0},\widehat{\mathbf{H}}_{k,1}\widehat{\mathbf{U}}^{(R)}_{1,2}\right] \in \mathbb{C}^{N \times M}$, while the second equality follows by noting that $M>N$ for $\frac{M}{N} \in \left[\frac{9K+\sqrt{81K^2+60K}}{30},\frac{3K+\sqrt{9K^2-12K}}{6}\right]$.
Moreover, as $M<KN$, we obtain
{\setlength{\abovedisplayskip}{4pt}
\setlength{\belowdisplayskip}{4pt}
\begin{small}
\begin{equation}
\mathrm{rank}\left(\sum^{K-1}_{k=0}\widehat{\mathbf{G}}_{0,k}\widehat{\mathbf{F}}_k\left[\widehat{\mathbf{W}}_{k,0},\widehat{\mathbf{H}}_{k,1}\widehat{\mathbf{U}}^{(R)}_{1,2}\right]\right) = \mathrm{min}\{KN,M\} =  M
\end{equation}
\end{small}}
\!\!with probability one, and so \eqref{prank} is met.

To summarize, the above constructed $\widehat{\mathbf{T}}$, $\{\widehat{\mathbf{U}}_{j,j'}\}$, and $\{\widehat{\mathbf{F}}_k\}$ satisfy \eqref{pairwisealignment}, \eqref{pneutralize}, and \eqref{prank} with probability one. Together with Lemma \ref{ProveLemma1}, we complete the proof of Lemma \ref{LemmaDoF1}.

\subsection{Proof of Lemma \ref{LemmaDoF2}}
The proof of Lemma \ref{LemmaDoF2} closely follows that of Lemma \ref{LemmaDoF1}. The main steps of the proof are presented as follows:
\begin{itemize}
\item Show that the signal alignment \eqref{alignment} can be performed and $d'$ in \eqref{d'} is well-defined.
\item For $\frac{M}{N} \in \Big[1, \frac{9K+\sqrt{81K^2+60K}}{30}\Big)$ and $d<\frac{3M^2+2MN}{9M+N}$, construct a set $\tilde{\mathcal{T}}_{M,N,d}$ such that a randomly generated channel realization belongs to $\tilde{\mathcal{T}}_{M,N,d}$ with probability one.
\item Prove that for almost all $\mathbf{T} \in \tilde{\mathcal{T}}_{M,N,d}$, there exist $\{\mathbf{U}_{j,j'}\}$ and $\{\tilde{\mathbf{F}}_k\}$ satisfying \eqref{alignment} and \eqref{pneutralize'}.
\end{itemize}

We first show that the signal alignment in \eqref{alignment} can be performed for $\frac{M}{N} \in \Big[1, \frac{9K+\sqrt{81K^2+60K}}{30}\Big)$ and $d<\frac{3M^2+2MN}{9M+N}$. From the discussion in Section IV-C, it suffices to show that $N'$ is well defined. That is, for $d$ close to $\frac{3M^2+2MN}{9M+N}$, $N'$ defined in \eqref{N'} satisfies $0\leq N' \leq N$. Too see this, note that 
{\setlength{\abovedisplayskip}{4pt}
\setlength{\belowdisplayskip}{4pt}
\begin{small}
\begin{equation}
N'  = \frac{2M-d}{K} > \frac{2M-\frac{3M^2+2MN}{9M+N}}{K} = \frac{15M^2}{K(9M+N)} > 0,
\end{equation}
\end{small}}
\!\!implying that $N'>0$ for any $d<\frac{3M^2+2MN}{9M+N}$.
Further, for $1 \leq \frac{M}{N}< \frac{9K+\sqrt{81K^2+60K}}{30}$, we have
{\setlength{\abovedisplayskip}{4pt}
\setlength{\belowdisplayskip}{4pt}
\begin{small}
\begin{equation}
\frac{2M-\frac{3M^2+2MN}{9M+N}}{K} < N.
\end{equation}
\end{small}}
\!\!\!Therefore, we can always choose $d$ close to $\frac{3M^2+2MN}{9M+N}$ to ensure $0< N' = \frac{2M-d}{K} \leq N$. As analogous to Appendix \ref{Proof of Lemma 1 and Lemma 2}, we can also verify that for $\frac{M}{N} \in \Big[1, \frac{9K+\sqrt{81K^2+60K}}{30}\Big)$, we can choose $d$ close to $\frac{3M^2+2MN}{9M+N}$ such that $0 \leq d' \leq d$. 

We now consider step 2 of the proof. For $\frac{M}{N} \in \Big[1, \frac{9K+\sqrt{81K^2+60K}}{30}\Big)$ and $d<\frac{3M^2+2MN}{9M+N}$, define
{\setlength{\abovedisplayskip}{4pt}
\setlength{\belowdisplayskip}{4pt}
\begin{equation}
\tilde{\mathcal{T}}_{M,N,d} = \left\{\mathbf{T}\quad\!\!\left|\quad\!\!
\begin{aligned}
& \text{All $\mathbf{T}$ satisfying:}\\[-0.3cm]
& \text{1) } \mathrm{rank}(\tilde{\mathbf{K}}_{j}) = KN',\forall j;\\[-0.3cm]
& \text{2) } \text{there exist $\{\mathbf{U}_{j,j'}\}$ satisfying \eqref{alignment} and $\mathrm{rank}(\tilde{\mathbf{K}}) = 3d'M$}
 \end{aligned}\right.\right\}
\end{equation}}
\!\!where $\tilde{\mathbf{K}}_j$ is defined in \eqref{K'_j} and
\begin{small}
\begin{align}
\label{K'}
\tilde{\mathbf{K}}=\left[
\begin{array}{cccc}
    \left(\tilde{\mathbf{H}}_{0,1}\mathbf{U}^{(L)}_{1,2}\right)^{T}\!\! \otimes \mathbf{G}_{0,0} &  \left(\tilde{\mathbf{H}}_{1,1}\mathbf{U}^{(L)}_{1,2}\right)^{T}\!\!\otimes \mathbf{G}_{0,1} & \cdots &  \left(\tilde{\mathbf{H}}_{K-1,1}\mathbf{U}^{(L)}_{1,2}\right)^{T}\!\!\otimes \mathbf{G}_{0,K-1}\\
       \left(\tilde{\mathbf{H}}_{2,0}\mathbf{U}^{(L)}_{2,0}\right)^{T}\!\!\otimes \mathbf{G}_{1,0} & \left(\tilde{\mathbf{H}}_{1,2}\mathbf{U}^{(L)}_{2,0}\right)^{T}\!\!\otimes \mathbf{G}_{1,1} & \cdots & \left(\tilde{\mathbf{H}}_{K-1,2}\mathbf{U}^{(L)}_{2,0}\right)^{T}\!\!\otimes \mathbf{G}_{1,K-1}\\
      \left(\tilde{\mathbf{H}}_{0,0}\mathbf{U}^{(L)}_{0,1}\right)^{T}\!\!\otimes \mathbf{G}_{2,0} & \left(\tilde{\mathbf{H}}_{1,0}\mathbf{U}^{(L)}_{0,1}\right)^{T}\!\!\otimes \mathbf{G}_{2,1} & \cdots & \left(\tilde{\mathbf{H}}_{K-1,0}\mathbf{U}^{(L)}_{0,1}\right)^{T}\!\!\otimes \mathbf{G}_{2,K-1}
\end{array}
\right]
\end{align}
\end{small}
\!\!\!where $\tilde{\mathbf{H}}_{k,j} = \mathbf{E}\mathbf{H}_{k,j}$ as defined in Section \ref{Signal Alignment II}. Similarly to Appendix \ref{Proof of Lemma 1}, we can verify that a random generated $\mathbf{T}$ belongs to $\tilde{\mathcal{T}}_{M,N,d}$ with probability one. 

We now consider the last step of the proof, i.e., to show that for almost all $\mathbf{T}$ in $\tilde{\mathcal{T}}_{M,N,d}$, there exist $\{\mathbf{U}_{j,j'}\}$ and $\{\tilde{\mathbf{F}}_k\}$ satisfying \eqref{alignment} and \eqref{pneutralize'}. As analogous to Lemma \ref{ProveLemma1}, we have the following result. Note that the proof of Lemma \ref{ProveLemma2} follows that of Lemma \ref{ProveLemma1} step by step, and is omitted for brevity.

\begin{lemma}
\label{ProveLemma2}
For $d<\frac{3M^2+2MN}{9M+N}$ and $\frac{M}{N} \in \Big[1, \frac{9K+\sqrt{81K^2+60K}}{30}\Big)$, assume that there exist a certain element $\widehat{\mathbf{T}} \in \tilde{\mathcal{T}}_{M,N,d}$, full-row-rank $\{\widehat{\mathbf{U}}_{j,j'}\}$, and relay processing matrices $\{\widehat{\tilde{\mathbf{F}}}_k\}$ such that \eqref{alignment} and \eqref{pneutralize'} hold.
Then for a random $\mathbf{T} \in \tilde{\mathcal{T}}_{M,N,d}$, there exist $\{\mathbf{U}_{j,j'}\}$ and $\{\tilde{\mathbf{F}}_k\}$ satisfying \eqref{alignment} and \eqref{pneutralize'} with probability one.
\end{lemma}

Based on Lemma 5, to show that for almost all $\mathbf{T} \in \tilde{\mathcal{T}}_{M,N,d}$ there exist $\{\mathbf{U}_{j,j'}\}$ and $\{\tilde{\mathbf{F}}_k\}$ satisfying \eqref{alignment} and \eqref{pneutralize'}, it suffices to find a certain $\widehat{\mathbf{T}} \in \tilde{\mathcal{T}}_{M,N,d}$ that satisfies the condition. To this end, we set 
{\setlength{\abovedisplayskip}{4pt}
\setlength{\belowdisplayskip}{4pt}
\begin{small}
\begin{align}
\label{F_tilde}
\widehat{\tilde{\mathbf{F}}}_k = [\mathbf{I}_{N'}, \mathbf{0}_{ N'\times (N-N')}]^T \in \mathbb{C}^{N \times N'}, \quad\!\! k=0,\cdots, K-1
\end{align}
\end{small}}
\!\!\!and randomly generate $\{\widehat{\mathbf{G}}_{j,k}\}$, $\{\widehat{\tilde{\mathbf{H}}}_{k,j}\widehat{\mathbf{U}}^{(R)}_{j,j+1}\}$ with the entries independently drawn from a continuous distribution. Then, we choose full-rank matrices $\{\widehat{\tilde{\mathbf{H}}}_{k,j}\widehat{\mathbf{U}}^{(L)}_{j,j+1}\}$ to satisfy
{\setlength{\abovedisplayskip}{4pt}
\setlength{\belowdisplayskip}{4pt}
\begin{small}
\begin{align}
\label{specificH'}
\underbrace{\left[\begin{array}{cccc}
    \widehat{\mathbf{G}}_{j,0}\widehat{\tilde{\mathbf{F}}}_0 &  \widehat{\mathbf{G}}_{j,1}\widehat{\tilde{\mathbf{F}}}_1 & \cdots & \widehat{\mathbf{G}}_{j,K-1}\widehat{\tilde{\mathbf{F}}}_{K-1}
\end{array}\right]}_{M \times KN'} \underbrace{\left[
\begin{array}{c}
\widehat{\tilde{\mathbf{H}}}_{0,j+1}\widehat{\mathbf{U}}^{(L)}_{j+1,j+2} \\ \vdots \\ \widehat{\tilde{\mathbf{H}}}_{K-1,j+1}\widehat{\mathbf{U}}^{(L)}_{j+1,j+2} \end{array}\!\!\right]}_{KN'\times d'}  = \mathbf{0}.
\end{align}
\end{small}}
\!\!\!With $\tilde{\mathbf{F}}_k$ in \eqref{F_tilde}, $\widehat{\mathbf{G}}_{j,k}\widehat{\tilde{\mathbf{F}}}_{k}$ is simply the first $N'$ columns of $\widehat{\mathbf{G}}_{j,k}$. For $\frac{M}{N} \in \Big[1, \frac{9K+\sqrt{81K^2+60K}}{30}\Big)$ and $d<\frac{3M^2+2MN}{9M+N}$, we have 
{\setlength{\abovedisplayskip}{3pt}
\setlength{\belowdisplayskip}{3pt}
\begin{equation}
KN'-M = M-d > d \geq d',
\end{equation}}
\!\!\!implying that the null space of the $M \times KN'$ matrix in \eqref{specificH'} has at least $d'$ dimensions. Thus full-rank $\{\widehat{\tilde{\mathbf{H}}}_{k,j}\widehat{\mathbf{U}}^{(L)}_{j,j+1}\}$ exist with probability one. Based on the chosen $\{\widehat{\tilde{\mathbf{H}}}_{k,j}\widehat{\mathbf{U}}^{(L)}_{j,j+1}\}$ and $\{\widehat{\tilde{\mathbf{H}}}_{k,j}\widehat{\mathbf{U}}^{(R)}_{j,j+1}\}$, we determine the values of $\{\widehat{\tilde{\mathbf{H}}}_{k,j}\}$ and $\{\widehat{\mathbf{U}}_{j,j+1}\}$ (not necessarily unique). With $\eqref{alignment}$, $\{\widehat{\mathbf{U}}_{j+1,j}\}$ are also determined. Finally, $\widehat{\mathbf{T}}$ is determined by $\{\widehat{\tilde{\mathbf{H}}}_{k,j}, \widehat{\mathbf{G}}_{j,k}\}$ and it can be verified that $\widehat{\mathbf{T}} \in \tilde{\mathcal{T}}_{M,N,d}$ with probability one.

We now show that the above constructed $\widehat{\mathbf{T}}$, $\{\widehat{\tilde{\mathbf{F}}}_{k}\}$, and $\{\widehat{\mathbf{U}}_{j,j'}\}$ satisfy \eqref{alignment} and \eqref{pneutralize'} with probability one. By construction, \eqref{alignment} is automatically met. Further, from \eqref{specificH'}, we see that \eqref{pneutralize1'} holds with probability one. To check \eqref{pneutralize4'}, it suffices to consider the case $j=0$ by symmetry. Note that $\{\widehat{\tilde{\mathbf{W}}}_{k,0}\}$ are determined by $\{\widehat{\tilde{\mathbf{H}}}_{k,2}\widehat{\mathbf{U}}^{(L)}_{2,0}\}$ and $\{\widehat{\tilde{\mathbf{H}}}_{k,0}\widehat{\mathbf{U}}^{(L)}_{0,1}\}$, which only depend on $\{\widehat{\mathbf{G}}_{1,k}\}$ and $\{\widehat{\mathbf{G}}_{2,k}\}$. That is, $\{\widehat{\tilde{\mathbf{W}}}_{k,0}\}$ are not functions of $\{\widehat{\mathbf{G}}_{0,k}\}$. Moreover, $\{\widehat{\mathbf{G}}_{0,k}\}$, $\{\widehat{\tilde{\mathbf{F}}}_k\}$, and $\{\widehat{\tilde{\mathbf{H}}}_{k,1}\widehat{\mathbf{U}}^{(R)}_{1,2}\}$ are independent of each other by construction. Therefore, 
{\setlength{\abovedisplayskip}{4pt}
\setlength{\belowdisplayskip}{4pt}
\begin{small}
\begin{equation}
\mathrm{rank}\left(\widehat{\mathbf{G}}_{0,k}\widehat{\tilde{\mathbf{F}}}_k\left[\widehat{\tilde{\mathbf{W}}}_{k,0},\widehat{\tilde{\mathbf{H}}}_{k,1}\widehat{\mathbf{U}}^{(R)}_{1,2}\right]\right) = \min\{M,N'\} = N'.
\end{equation}
\end{small}}
\!\!As $M<KN' = 2M-d$, we obtain
{\setlength{\abovedisplayskip}{4pt}
\setlength{\belowdisplayskip}{4pt}
\begin{small}
\begin{equation}
\mathrm{rank}\left(\sum^{K-1}_{k=0}\widehat{\mathbf{G}}_{0,k}\widehat{\tilde{\mathbf{F}}}_k\left[\widehat{\tilde{\mathbf{W}}}_{k,0},\widehat{\tilde{\mathbf{H}}}_{k,1}\widehat{\mathbf{U}}^{(R)}_{1,2}\right]\right) = \mathrm{min}\{KN',M\} =  M
\end{equation}
\end{small}}
\!\!with probability one, and so \eqref{pneutralize4'} is met. We complete the proof of Lemma \ref{LemmaDoF2}.

\subsection{Proof of Lemma \ref{LemmaDoF3}}
As analogous to the arguments in Appendix I-A, the general solution of \eqref{zeroforcing_1}-\eqref{zeroforcing_2} with respect to $\{\mathbf{F}_k\}$ can be represented by finite-degree polynomials of the entries of $\{\mathbf{H}_{k,j}, \mathbf{G}_{j,k}\}$. Then similar to Lemma \ref{ProveLemma1}, we have the following lemma for the solvability of \eqref{zeroforcing}. Note that to ease the notation, we do not define the set $\mathcal{T}_{M,N,d}$, but rephrase its constraints as conditions in the lemma. We omit the proof for brevity. 

\begin{lemma}
\label{ProveLemma3}
For $\frac{M}{N} \in \Big(0,\frac{\sqrt{3K}}{K}\Big)$ and $d = \frac{M}{2}$, assume that there exists a certain choice of $\{\widehat{\mathbf{H}}_{k,j},\widehat{\mathbf{G}}_{j,k}\}$ and $\{\widehat{\mathbf{F}}_k\}$ satisfying \eqref{zeroforcing} and
$\mathrm{rank}(\widehat{\mathbf{K}}') = 3M^2$, where $\widehat{\mathbf{K}}'$, similar to $\mathbf{K}$, is defined by rewriting \eqref{zeroforcing_1}-\eqref{zeroforcing_2} using Kronecker product, i.e.,
{\setlength{\abovedisplayskip}{4pt}
\setlength{\belowdisplayskip}{4pt}
\begin{small}
\begin{align}
\label{Kbar}
\widehat{\mathbf{K}}'\!=\!\left[\!\!
\begin{array}{cccc}
    \left(\widehat{\mathbf{H}}_{0,1}\mathbf{U}_{1,2}\right)^{T}\!\!\!\! \otimes \widehat{\mathbf{G}}_{0,0}\! & \left(\widehat{\mathbf{H}}_{1,1}\mathbf{U}_{1,2}\right)^{T}\!\!\!\!\otimes \widehat{\mathbf{G}}_{0,1} & \!\cdots\! &  \left(\widehat{\mathbf{H}}_{K-1,1}\mathbf{U}_{1,2}\right)^{T}\!\!\!\!\otimes \widehat{\mathbf{G}}_{0,K-1}\\
    \left(\widehat{\mathbf{H}}_{0,2}\mathbf{U}_{2,1}\right)^{T}\!\!\!\! \otimes \widehat{\mathbf{G}}_{0,0}\! & \left(\widehat{\mathbf{H}}_{1,2}\mathbf{U}_{2,1}\right)^{T}\!\!\!\!\otimes \widehat{\mathbf{G}}_{0,1} & \!\cdots\! &  \left(\widehat{\mathbf{H}}_{K-1,2}\mathbf{U}_{2,1}\right)^{T}\!\!\!\!\otimes \widehat{\mathbf{G}}_{0,K-1}\\
    \left(\widehat{\mathbf{H}}_{0,2}\mathbf{U}_{2,0}\right)^{T}\!\!\!\! \otimes \widehat{\mathbf{G}}_{1,0}\! & \left(\widehat{\mathbf{H}}_{1,2}\mathbf{U}_{2,0}\right)^{T}\!\!\!\!\otimes \widehat{\mathbf{G}}_{1,1} & \!\cdots\! &  \left(\widehat{\mathbf{H}}_{K-1,2}\mathbf{U}_{2,0}\right)^{T}\!\!\!\!\otimes \widehat{\mathbf{G}}_{1,K-1}\\
    \left(\widehat{\mathbf{H}}_{0,0}\mathbf{U}_{0,2}\right)^{T}\!\!\!\! \otimes \widehat{\mathbf{G}}_{1,0}\! & \left(\widehat{\mathbf{H}}_{1,0}\mathbf{U}_{0,2}\right)^{T}\!\!\!\!\otimes \widehat{\mathbf{G}}_{1,1} & \!\cdots\! &  \left(\widehat{\mathbf{H}}_{K-1,0}\mathbf{U}_{0,2}\right)^{T}\!\!\!\!\otimes \widehat{\mathbf{G}}_{1,K-1}\\
    \left(\widehat{\mathbf{H}}_{0,0}\mathbf{U}_{0,1}\right)^{T}\!\!\!\! \otimes \widehat{\mathbf{G}}_{2,0}\! & \left(\widehat{\mathbf{H}}_{1,0}\mathbf{U}_{0,1}\right)^{T}\!\!\!\!\otimes \widehat{\mathbf{G}}_{2,1} & \!\cdots\! &  \left(\widehat{\mathbf{H}}_{K-1,0}\mathbf{U}_{0,1}\right)^{T}\!\!\!\!\otimes \widehat{\mathbf{G}}_{2,K-1}\\
      \left(\widehat{\mathbf{H}}_{0,1}\mathbf{U}_{1,0}\right)^{T}\!\!\!\!\otimes \widehat{\mathbf{G}}_{2,0} \!& \left(\widehat{\mathbf{H}}_{1,1}\mathbf{U}_{1,0}\right)^{T}\!\!\!\!\otimes \widehat{\mathbf{G}}_{2,1} & \!\cdots\! & \left(\widehat{\mathbf{H}}_{K-1,1}\mathbf{U}_{1,0}\right)^{T}\!\!\!\!\otimes \widehat{\mathbf{G}}_{2,K-1}
\end{array}
\!\!\right]\!\in\! \mathbb{C}^{3M^2\times KN^2}\!.
\end{align}
\end{small}}
\!\!Then for random $\{\mathbf{H}_{k,j},\mathbf{G}_{j,k}\}$, together with $\{\mathbf{U}_{j,j'}\}$ and $\{\mathbf{V}_j\}$ constructed in \eqref{UV}, there exist $\{\mathbf{F}_k\}$ satisfying \eqref{zeroforcing} with probability one.
\end{lemma}

Based on Lemma 6, what remains is to construct $\{\widehat{\mathbf{H}}_{k,j},\widehat{\mathbf{G}}_{j,k}\}$ and $\{\widehat{\mathbf{F}}_k\}$ satisfying \eqref{zeroforcing} and $\mathrm{rank}(\widehat{\mathbf{K}}') = 3M^2$. We basically follow the construction in the proof of Lemma \ref{LemmaDoF1}. Let $\widehat{\mathbf{F}}_k = \mathbf{I}_k$, for $k=0,\cdots,K-1$, and $\{\widehat{\mathbf{G}}_{j,k}\}$ be random matrices with entries independently drawn from a continuous distribution. Then we design $\{\widehat{\mathbf{H}}_{k,j}\mathbf{U}_{j,j'}\}$ with full column rank satisfying
{\setlength{\abovedisplayskip}{4pt}
\setlength{\belowdisplayskip}{4pt}
\begin{small}
\begin{align}
\label{specificH''}
\underbrace{\left[\begin{array}{cccc}
    \widehat{\mathbf{G}}_{j,0} &  \widehat{\mathbf{G}}_{j,1} & \cdots & \widehat{\mathbf{G}}_{j,K-1}
\end{array}\right]}_{M \times KN}
 \underbrace{\left[\begin{array}{cc}
 \widehat{\mathbf{H}}_{1,j+1}\mathbf{U}_{j+1,j+2} &\widehat{\mathbf{H}}_{1,j+2}\mathbf{U}_{j+2,j+1}\\  
 \vdots &\vdots\\ 
 \widehat{\mathbf{H}}_{K-1,j+1}\mathbf{U}_{j+1,j+2} & \widehat{\mathbf{H}}_{K-1,j+2}\mathbf{U}_{j+2,j+1} \end{array}\!\!\right]}_{KN\times M}  = \mathbf{0}, \quad\!\! j = 0,1,2
\end{align}
\end{small}}
\!\!\!where $\{\mathbf{U}_{j,j'}\}$ are constructed in \eqref{UV}. From the rank-nullity theorem, there exist full-column-rank $\{\widehat{\mathbf{H}}_{k,j}\mathbf{U}_{j,j'}\}$ satisfying \eqref{specificH''} with probability one, provided $2M \leq KN$, \Big(which is true for $\frac{M}{N} \in \left(0, \frac{\sqrt{3K}}{3} \right)$\Big). Note that from \eqref{UV}, $\widehat{\mathbf{H}}_{k,j}\mathbf{U}_{j,j+1}$ and $\widehat{\mathbf{H}}_{k,j+1}\mathbf{U}_{j+1,j}$ are the left $\frac{M}{2}$ columns and the right $\frac{M}{2}$ columns of $\widehat{\mathbf{H}}_{k,j}$, respectively. Thus, we can combine them together to form $\widehat{\mathbf{H}}_{k,j}$ with full column rank. It can be readily verified that, together with $\{\mathbf{V}_j\}$ and $\{\mathbf{U}_{j,j'}\}$ in \eqref{UV}, the above constructed $\{\widehat{\mathbf{H}}_{k,j}, \widehat{\mathbf{G}}_{j,k}\}$ and $\{\widehat{\mathbf{F}}_k\}$ satisfy \eqref{zeroforcing_1}-\eqref{zeroforcing_2}. Moreover, since $\{\widehat{\mathbf{G}}_{j,k}\}$ are randomly generated, \eqref{rank} and $\mathrm{rank}(\widehat{\mathbf{K}}') = 3M^2$ are met with probability one, which concludes the proof of Lemma \ref{LemmaDoF3}.

\section{Symbol Extension}
With the techniques of symbol extension and antenna disablement, we show that all the results in Section IV hold for any rational DoF $d$. We focus on Signal Alignment I in Section IV-B. The treatments for other signal alignment approaches are similar. Let $M^*<M$ satisfying $\frac{M^*}{N} \in \Big[ \frac{9K+\sqrt{81K^2+60K}}{30},$ $\frac{3K+\sqrt{9K^2-12K}}{6} \Big)$. Further, assume that both $M^*$ and $d$ are rational numbers satisfying $d<\frac{M^*}{3}+\frac{KN^2}{9M^*}$. Let $L$ be a positive integer such that both $LM^*$ and $Ld$ are integers. 
Then, we extend the channel by $L$ times and disable $(M-M^*)L$ antennas at each user end. With some abuse of notation, we represent the extended channel of each link in a block-diagonal form as 
{\setlength{\abovedisplayskip}{1pt}
\setlength{\belowdisplayskip}{1pt}
\begin{subequations}
\begin{small}
\begin{align}
\label{block}
\mathbf{H}_{k,j} &=\mathrm{diag}\left(\mathbf{H}^{(0)}_{k,j},\mathbf{H}^{(1)}_{k,j},\cdots,\mathbf{H}^{(L-1)}_{k,j}\right)\in \mathbb{C}^{LN \times LM^*}\\
\mathbf{G}_{j,k} &= \mathrm{diag}\left(\mathbf{G}^{(0)}_{j,k}, \mathbf{G}^{(1)}_{j,k} \cdots,\mathbf{G}^{(L-1)}_{j,k}\right)\in \mathbb{C}^{LM^* \times LN}
\end{align}
\end{small}
\end{subequations}}
\!\!\!\!where $\mathbf{H}^{(l)}_{k,j}$ and $\mathbf{G}^{(l)}_{j,k}$ are the uplink and downlink channel matrices in the $(l+1)$-th channel use respectively, with
\begin{align}
&\mathbf{H}^{(l)}_{k,j} \in \mathbb{C}^{N \times \lceil M^* \rceil}, \mathbf{G}^{(l)}_{j,k} \in \mathbb{C}^{\lceil M^* \rceil \times N},\quad \!\! \mathrm{for} \quad\!\! l = 0, \cdots, L\left(M^* -\lfloor M^*\rfloor\right)-1\nonumber \\
&\mathbf{H}^{(l)}_{k,j} \in \mathbb{C}^{N \times \lfloor M^* \rfloor},  \mathbf{G}^{(l)}_{j,k} \in \mathbb{C}^{\lfloor M^* \rfloor \times N}, \quad\!\! \mathrm{for} \quad\!\! l=L\left(M^* -\lfloor M^*\rfloor\right), \cdots, L-1. \nonumber
\end{align}

Our goal is to show that any DoF $Ld$ with $d < \frac{M^*}{3}+\frac{KN^2}{9M^*}$ is achievable for the extended channel, implying that an average DoF $d$ is achieved per channel use. To this end, we need to design $\mathbf{F}_k \in \mathbb{C}^{LN \times LN}, \mathbf{U}_{j,j'} \in \mathbb{C}^{LM^* \times Ld}$ and $\mathbf{V}_{j} \in \mathbb{C}^{2Ld \times LM^*}$ to satisfy condition \eqref{zeroforcing} (with the rank $2d$ in \eqref{rank} replaced by $2Ld$). This can be done by performing Signal Alignment I and following the arguments in Appendix I-A step by step, with the only main difference being that $\{\mathbf{H}_{k,j}\}$ and $\{\mathbf{G}_{j,k}\}$ take the block diagonal forms in \eqref{block}, which does not change our conclusion. We omit the details for brevity.
\end{appendices}

\ifCLASSOPTIONcaptionsoff
  \newpage
\fi

\end{document}